\documentclass[a4paper]{article}
\usepackage[cp1251]{inputenc}
\usepackage[T2A]{fontenc}
\usepackage[english]{babel}

\usepackage{ccaption}
\captiondelim{. }

\usepackage[dvips]{graphicx}

\usepackage{amsmath,amssymb,amsthm}
\usepackage{hyperref}

\usepackage[top=20mm,right=20mm,left=20mm,bottom=20mm]{geometry}

\def\ds{\displaystyle}
\def\bbR{\mathbb{R}}
\def\gs{\geqslant}
\def\ls{\leqslant}
\def\ds{\displaystyle}
\def\mfT{{J}}
\def\vm{m}

\theoremstyle{plain}
\newtheorem{propos}{Proposition}

\begin{document}


\title{Phase topology of one integrable case of the rigid body motion%
\footnote{Submitted on August 17, 1977.}}

\author{M.P. Kharlamov\footnote{Donetsk Physico-technical Institute.}}

\date{}

\maketitle

\begin{center}
{\bf \textit{Mekh. Tverd. Tela} (Russian Journal ``Mechanics of Rigid Body''),\\
1979, No. 11, pp. 50--64.}

\vspace{5mm}

\href{http://www.ics.org.ru/doc?pdf=1156\&dir=r}{http://www.ics.org.ru}

\vspace{2mm}

\href{https://www.researchgate.net/publication/252860538\_Phase\_topology\_of\_one\_integrable\_case\_of\_the\_motion\_of\_a\_rigid\_body?ev=prf\_pub}{https://www.researchgate.net}

\end{center}

\vspace{3mm}

{
\small The paper is devoted to the investigation of the motion of a rigid body fixed at its mass center under the influence of a central Newtonian field with the condition that the constant of the area integral is zero.

\vspace{3mm}

}

Consider a body fixed by its mass center $O$ at the origin of an immovable in space reference frame $O\vm_1\vm_2\vm_3$ in such a way that the orth $\vm_1$ is directed to the center of the gravitational attraction. We attach to the body the basis $e_1e_2e_3$ of the unit direction vectors of the central principal inertia axes. The corresponding inertia moments are denoted by $I_1, I_2, I_3$ and the inertia ellipsoid is considered triaxial:
\begin{equation}\label{neq1}
  I_1 < I_2 < I_3.
\end{equation}

Let $\alpha_1, \alpha_2, \alpha_3$ be the components of $\vm_1$ in the moving basis,
\begin{equation}\label{neq2}
  \alpha_1^2+\alpha_2^2+\alpha_3^2=1,
\end{equation}
and $\omega_1, \omega_2, \omega_3$ the projections onto the moving axes of the instant angular velocity vector.

The potential energy $V$ of the considered system is preserved by the rotations of the body around the axis $O\vm_1$, consequently (see \cite{bib05}) $V=\tilde{V} \circ p$, where $\tilde{V} = \tilde{V}(\alpha_1, \alpha_2, \alpha_3)$ is a function on the sphere \eqref{neq2} and $p : SO(3) \rightarrow S^2$ is the projection of the configuration space onto the Poisson sphere \eqref{neq2}; $p$ takes the direction cosines matrix of the immovable basis with respect to the movable one to its first row. The generally accepted approximation of the potential is
\begin{equation}\label{neq3}
  \tilde{V}=\frac{\varepsilon^2}{2I_1I_2I_3}(I_1\alpha_1^2+I_2\alpha_2^2+I_3\alpha_3^2).
\end{equation}
The parameter $\varepsilon \neq 0$ depends on the interaction characteristics. It can be excluded by a linear change of the coordinates and the time, therefore we put $\varepsilon = 1$.

In addition to the classical integrals of the energy
\begin{equation}\label{neq4}
  H=\frac{1}{2}(I_1\omega_1^2+I_2\omega_2^2+I_3\omega_3^2)+\frac{1}{2I_1I_2I_3}(I_1\alpha_1^2+I_2\alpha_2^2+I_3\alpha_3^2)
\end{equation}
and momentum
\begin{equation*}
  L=I_1\alpha_1\omega_1+I_2\alpha_2\omega_2+I_3\alpha_3\omega_3,
\end{equation*}
the system admits the following quadratic integral
\begin{equation}\label{neq5}
  F=\frac{1}{2}(I_1\omega_1^2+I_2\omega_2^2+I_3\omega_3^2)-\frac{1}{2I_1I_2I_3}(I_2I_3\alpha_1^2+I_3I_1\alpha_2^2+I_1I_2\alpha_3^2).
\end{equation}
This integral was first pointed out by A.\,Clebsch \cite{bib09} for the problem of the motion of a rigid body in fluid; then $\alpha_1, \alpha_2, \alpha_3$ played the role of the impulse force components. Later this integral was mentioned in the work of F.\,Tisserand \cite{bib10} considering this problem as describing the motion of a rigid body on the Earth surface (though Tisserand did not take into account the centrifugal forces restricting himself to the Coriolis forces).

Obviously, the function \eqref{neq5} is preserved by the transformations of the phase space generated by rotations around the $O\vm_1$ axis. Consequently, it is in involution with the momentum integral in the Lagrangian symplectic structure on $T(SO(3))$. Therefore the mechanical system describing the rigid body motion about its mass center in the field \eqref{neq3} is completely integrable \cite{bib02}. In the general case the analytical solution of the problem is complicated. It is based on finding the roots of the fourth degree polynomial \cite{bib07} and does not give a clear picture of motion. We investigate a partial case, namely, those motions of the rigid body which satisfy the equality
\begin{equation}\label{neq6}
  L=0.
\end{equation}

S.A.\,Chaplygin \cite{bib08} was the first who pointed out that under condition \eqref{neq6} the variables separate and he gave the generating function for the corresponding canonical transformation. Under the same condition in the work \cite{bib07} the geometrical interpretation of motions was given generalizing that of Poinsot. Nevertheless, the description of the types of motions has not been given up-to-date. In what follows, we complete this task.

{\bf Separation of variables.} To solve the problem, we will use the reduced system. According to \cite{bib06, bib07}, the reduced system is a natural mechanical system on the Poisson sphere. The functions \eqref{neq4} and \eqref{neq5} induce the first integrals $\tilde{H}$ and $\tilde{F}$ of the reduced system; these integrals are in involution with respect to the Lagrangian symplectic structure on $T(S^2)$ defined by the reduced metric (see \cite{bib05}):
\begin{equation}\label{neq7}
  \{\tilde{H}, \tilde{F}\} \equiv 0.
\end{equation}

Using equations (68) of the work \cite{bib05}, we write the integrals $\tilde{H}$ and $\tilde{F}$ in the form
\begin{eqnarray}
& &  2\tilde{H} = \ds{\frac{a_1\dot{\alpha}_1^2 + a_2\dot{\alpha}_2^2 + a_3\dot{\alpha}_3^2}{a_2a_3\alpha_1^2 + a_3a_1\alpha_2^2 + a_1a_2\alpha_3^2}} + a_2a_3\alpha_1^2 + a_3a_1\alpha_2^2 + a_1a_2\alpha_3^2 = h, \label{neq8} \\
& &  2\tilde{F} = \ds{\frac{a_1^2\dot{\alpha}_1^2 + a_2^2\dot{\alpha}_2^2 + a_3^2\dot{\alpha}_3^2 - \left(a_1\alpha_1\dot{\alpha}_1 + a_2\alpha_2\dot{\alpha}_2 + a_3\alpha_3\dot{\alpha}_3\right)^2}{\left(a_2a_3\alpha_1^2 + a_3a_1\alpha_2^2 + a_1a_2\alpha_3^2\right)^2}} - \label{neq9} \\
& & \phantom{2\tilde{F} =}  - \left(a_1\alpha_1^2 + a_2\alpha_2^2 + a_3\alpha_3^2\right) = f,\nonumber
\end{eqnarray}
where $h, f$ are the corresponding constants and
\begin{equation}\label{neq10}
  a_1 = \frac{1}{I_1}, \quad a_2 = \frac{1}{I_2}, \quad a_3 = \frac{1}{I_3}.
\end{equation}

The quadratic relations \eqref{neq8}, \eqref{neq9} guarantee that the reduced system belongs to the Liouville type \cite{bib03}, hence the separation of variables can be obtained in a purely algebraic way. Let us introduce the elliptic coordinates $\lambda, \mu$ by the formulas
\begin{equation}\label{neq11}
  \alpha_1^2=\frac{(a_1-\lambda)(a_1-\mu)}{(a_1-a_2)(a_1-a_3)}, \quad \alpha_2^2=\frac{(\lambda-a_2)(a_2-\mu)}{(a_1-a_2)(a_2-a_3)}, \quad \alpha_3^2=\frac{(\lambda-a_3)(\mu-a_3)}{(a_1-a_3)(a_2-a_3)}.
\end{equation}
According to \eqref{neq1}, \eqref{neq10} the domain for $\lambda$ and $\mu$ is
\begin{equation}\label{neq12}
  a_3 \ls \mu \ls a_2 \ls \lambda \ls a_1.
\end{equation}
From \eqref{neq8}, \eqref{neq9}, and \eqref{neq11} we obtain
\begin{equation}\label{neq13}
\begin{array}{c}
    \ds{h = \frac{\lambda - \mu}{4\lambda\mu}\left[\frac{\lambda\dot{\lambda}^2}{\varphi(\lambda)} - \frac{\mu\dot{\mu}^2}{\varphi(\mu)} \right] + \lambda\mu}, \quad
    \ds{f = \frac{\lambda - \mu}{4\lambda^2\mu^2}\left[\frac{\lambda^2\dot{\lambda}^2}{\varphi(\lambda)} - \frac{\mu^2\dot{\mu}^2}{\varphi(\mu)} \right] + (\lambda + \mu) - (a_1 + a_2 + a_3),}
\end{array}
\end{equation}
where $\varphi(s) = (a_1 - s)(a_2 - s)(a_3 - s)$ and according to \eqref{neq12} $\varphi(\lambda) \gs 0$, $\varphi(\mu) \ls 0$. Let us put
\begin{equation}\label{neq14}
  \xi=\int\limits_{a_2}^\lambda\sqrt{\frac{\lambda}{\varphi(\lambda)}}d\lambda, \qquad \eta=\int\limits_{a_3}^\mu\sqrt{-\frac{\mu}{\varphi(\mu)}}d\mu.
\end{equation}
Then
\begin{equation*}
  0 \ls \xi \ls m = \int\limits_{a_2}^{a_1}\sqrt{\frac{\lambda}{\varphi(\lambda)}}d\lambda, \quad 0 \ls \eta \ls n = \int\limits_{a_3}^{a_2}\sqrt{-\frac{\mu}{\varphi(\mu)}}d\mu.
\end{equation*}
Let us make the change of variables \eqref{neq14} in \eqref{neq13} and introduce the ``amended time'' $\tau$,
\begin{equation}\label{neq15}
  dt = \frac{\lambda - \mu}{2\lambda\mu}d\tau.
\end{equation}
Solving \eqref{neq13} with respect to $d\xi/d\tau, d\eta/d\tau$, we get
\begin{equation}\label{neq16}
\begin{array}{c}
  \ds{\left( \frac{d\xi}{d\tau} \right)^2 = f - \frac{h}{\lambda(\xi)}+[a_1+a_2+a_3-\lambda(\xi)]}, \quad
  \ds{\left( \frac{d\eta}{d\tau} \right)^2 = \frac{h}{\mu(\eta)} - f -[a_1+a_2+a_3-\mu(\eta)]}.
\end{array}
\end{equation}
Thus, on any fixed integral manifold \eqref{neq8}, \eqref{neq9} the variables are separated and the trajectories of the reduced system are defined by equations \eqref{neq16}.

{\bf Bifurcation diagram.} In the work \cite{bib01}, the investigation of motions in Liouville systems, in particular, in the problem of two immovable centers, is fulfilled. The method of such investigation is based on the study of the system of equations determining the so-called multiple roots curve. In the case considered, the multiple roots curve in fact coincides with the set of the pairs of values of the integrals \eqref{neq8}, \eqref{neq9} for which the functions $\tilde{H}$ and $\tilde{F}$ are dependent.

Let us introduce the integral map
\begin{equation}\label{neq17}
  \mfT=2\tilde{H} \times 2\tilde{F} : T(S^2) \to \bbR^2.
\end{equation}
From relations \eqref{neq8}, \eqref{neq9}, it follows that for all $(h,f) \in \bbR^2$ the sets
\begin{equation}\label{neq18}
  \mfT_{h,f}=\left\{(\nu, \dot{\nu}) \in T(S^2):\mfT(\nu,\dot{\nu})=(h,f) \right\}
\end{equation}
are integral manifolds of the reduced system (typically two-dimensional), i.e., the surfaces built by the trajectories of the corresponding vector field $\tilde{X}_0$ \cite{bib05}.

In most of the cases it is easy to answer the question on the topological structure of the manifold \eqref{neq18} and on the type of motion on it. Let us define the bifurcation set $\Sigma$ as the set of points $(h,f) \in \bbR^2$ over which the map \eqref{neq17} is not locally trivial. Since the levels of the functions $\tilde{H}$ and $\tilde{F}$ are compact sets, the set $\Sigma$ is the set of critical values of the map $\mfT$. Then the Liouville theorem \cite{bib02} and equation \eqref{neq7} yield the following statement.

\begin{propos}\label{prop1} Suppose $(h,f) \in \bbR^2\setminus\Sigma$. Then either $\mfT_{h,f} = \varnothing$, or each connected component of $\mfT_{h,f}$ is diffeomorphic to the two-dimensional torus and the motion on it is quasi-periodic. The corresponding frequencies are found by integrating equations \eqref{neq16}.
\end{propos}

In fact, the values $(h,f) \in \Sigma$ are of the main interest. First, because the corresponding motions are not quasi-periodic and have a non-trivial character. Second, the form of the integral surfaces (not manifolds any more) $\mfT_{h,f}$ defines the transformations in the families of the integral tori.

We will find the set $\Sigma$ with the help of the redundant coordinates $\alpha_1$, $\alpha_2$, $\alpha_3$, $\dot{\alpha}_1$, $\dot{\alpha}_2$, $\dot{\alpha}_3$ on the tangent bundle to the Poisson sphere. Then we have to take into account the relation
\begin{equation}\label{neq19}
  \Gamma=\alpha_1^2+\alpha_2^2+\alpha_3^2=1.
\end{equation}
We also keep in mind the equality
\begin{equation}\label{neq20}
  \alpha_1\dot{\alpha}_1+\alpha_2\dot{\alpha}_2+\alpha_3\dot{\alpha}_3=0,
\end{equation}
but as a restriction it appears to be not essential. Let us compose a function with undefined Lagrange multipliers $\Lambda=\lambda_F\tilde{F} + \lambda_H\tilde{H} + \lambda_\Gamma\tilde{\Gamma}$.
Critical points of the map \eqref{neq17} satisfy equations \eqref{neq19}, \eqref{neq20} and the system
\begin{equation*}
  \Lambda_{\alpha_1}=\Lambda_{\alpha_2}=\Lambda_{\alpha_3}=0,\quad \Lambda_{\dot{\alpha}_1}=\Lambda_{\dot{\alpha}_2}=\Lambda_{\dot{\alpha}_3}=0.
\end{equation*}
In more detailed form, after some simplifying we get
\begin{equation}\label{neq21}
\begin{array}{ll}
  \ds{\frac{\lambda_F}{D}[a_1\dot{\alpha}_1-(a_1\alpha_1\dot{\alpha}_1+a_2\alpha_2\dot{\alpha}_2+a_3\alpha_3\dot{\alpha}_3)]-\lambda_H\dot{\alpha}_1=0,} & \\
  \ds{\lambda_F\frac{a_1\dot{\alpha}_1}{D^2}(a_1\alpha_1\dot{\alpha}_1+a_2\alpha_2\dot{\alpha}_2+a_3\alpha_3\dot{\alpha}_3) - } & (123)\\
\qquad   \ds{-\lambda_H\frac{a_2a_3\alpha_1}{D^3}(a_1\dot{\alpha}_1^2+a_2\dot{\alpha}_2^2+a_3\dot{\alpha}_3^2)+(\lambda_F a_1-\lambda_H a_2a_3)\alpha_1-2\lambda_\Gamma\alpha_1=0}, &
\end{array}
\end{equation}
where $D=a_2a_3\alpha_1^2+a_3a_1\alpha_2^2+a_1a_2\alpha_3^2$, and the symbol $(123)$ shows that the equations not written explicitly here are obtained by the cyclic transpositions of the indices.
Substituting the solutions of system \eqref{neq19} -- \eqref{neq21} into \eqref{neq8}, \eqref{neq9}, we find the set $\Sigma$ as three half-lines
\begin{eqnarray}
  h=a_1(f+a_2+a_3), &{}& f \gs -a_2,\label{neq22}\\
  h=a_2(f+a_3+a_1), &{}& f \gs -a_1,\label{neq23}\\
  h=a_3(f+a_1+a_2), &{}& f \gs -a_3 \label{neq24}
\end{eqnarray}
and the parabolic segment
\begin{equation}\label{neq25}
\begin{array}{cc}
  h=\ds{\frac{(f+a_1+a_2+a_3)^2}{4}}, & a_2-(a_3+a_1) \ls f \ls a_1 - (a_2 + a_3),
\end{array}
\end{equation}
the ends of which are tangent points of the parabola and the half-lines \eqref{neq23}, \eqref{neq22} (see Fig.~\ref{fig_01}).

\begin{figure}[!ht]
\centering
\includegraphics[width=0.35\paperwidth,keepaspectratio]{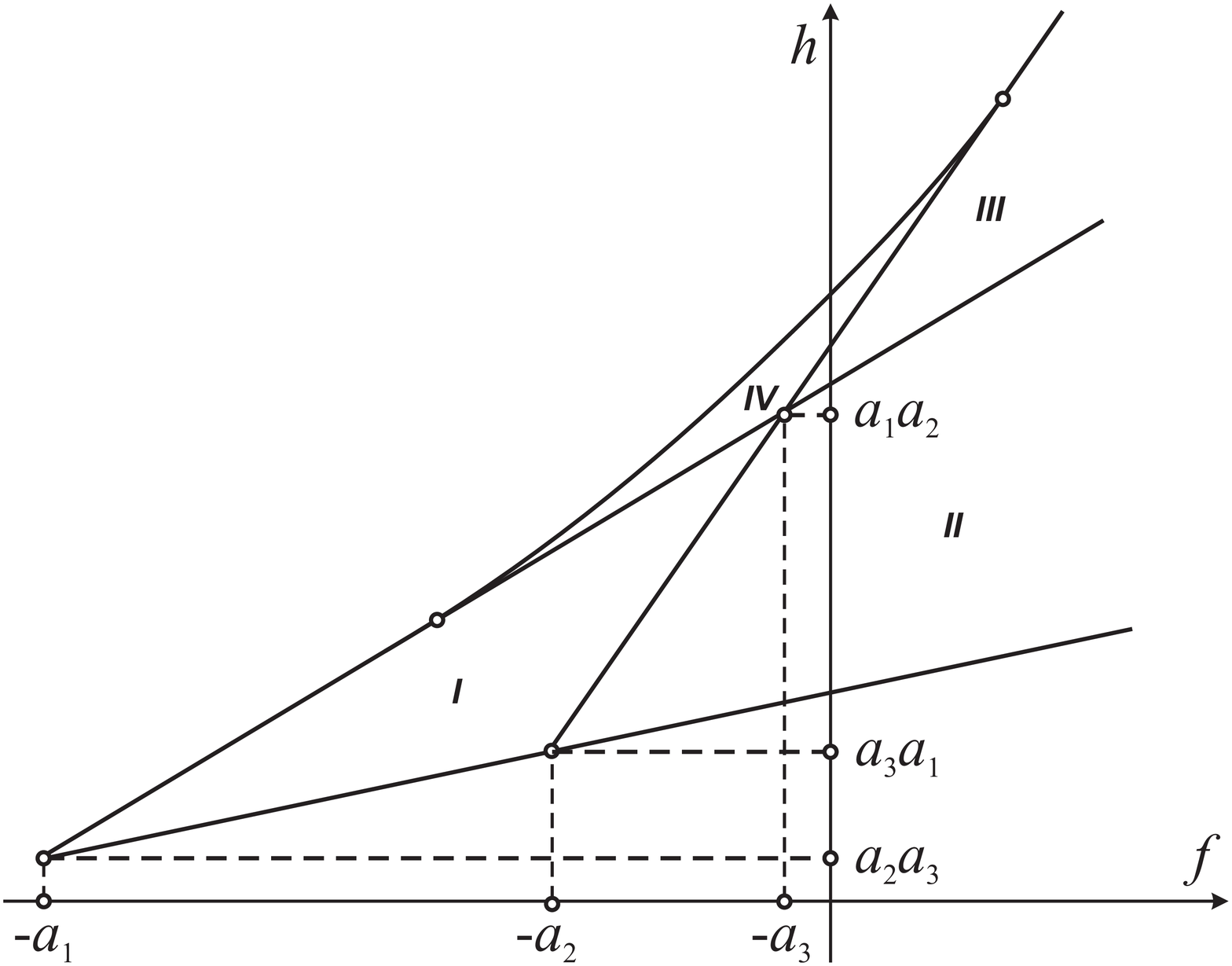}
\caption{}\label{fig_01}
\end{figure}

{\bf Types of critical points and motions.} According to \eqref{neq16}, for the fixed values $h$ and $f$ the coordinates of the vector $\vm_1$ on the Poisson sphere satisfy the inequalities
\begin{equation}\label{neq26}
  \ds{\frac{h}{\lambda}}+\lambda \ls f+a_1+a_2+a_3 \ls \ds{\frac{h}{\mu}}+\mu.
\end{equation}
Therefore the set $\tilde{M}_{h,f} \subset S^2$ defined by inequalities \eqref{neq26} is called the region of possible motions.

Let us denote
\begin{equation}\label{neq27}
\begin{array}{c}
  2f_0=f+a_1+a_2+a_3, \\
  f_1=f+a_2+a_3, \qquad f_2=f+a_3+a_1, \qquad  f_3=f+a_1+a_2,\\
  h_1=a_1 f_1, \qquad h_2=a_2 f_2, \qquad h_3=a_3 f_3.
\end{array}
\end{equation}

We start the investigation of critical cases from the values
\begin{equation}\label{neq28}
  f=-a_1, \qquad  h=a_2a_3.
\end{equation}
Substituting them into \eqref{neq16}, we get
\begin{equation*}
  \left(\frac{d\xi}{d\tau}\right)^2=-\frac{(\lambda-a_2)(\lambda-a_3)}{\lambda}, \qquad \left(\frac{d\eta}{d\tau}\right)^2=-\frac{(a_2-\mu)(\mu-a_3)}{\mu}.
\end{equation*}
Thus according to \eqref{neq12} and \eqref{neq14}, the only solution with \eqref{neq28} is $\xi=\eta\equiv0$. So the set $\mfT_{a_2a_3,-a_1}$ consists of two equilibrium points $(\pm 1,0,0)$.

Let the point $(h, f)$ belong to the segment of \eqref{neq24} with
\begin{equation}\label{neq29}
  h=h_3(f), \qquad  -a_1 < f < -a_2.
\end{equation}
From \eqref{neq16} we have
\begin{equation}\label{neq30}
  \left(\frac{d\xi}{d\tau}\right)^2=\frac{(\lambda-a_3)(f_3-\lambda)}{\lambda}, \qquad \left(\frac{d\eta}{d\tau}\right)^2=-\frac{(\mu-a_3)(f_3-\mu)}{\mu},
\end{equation}
and due to \eqref{neq27} and \eqref{neq29} $a_2<f_3<a_3$. During the motion $\eta \equiv 0$; the corresponding trajectories on the sphere are oscillations along the cross section $\alpha_3=0$ about the points $(\pm 1, 0, 0)$ in the limits $\lambda \ls f_3$. The manifold $\mfT_{h_2,f}$ consists of two non-intersecting circles.

For the values $h=h_3(f)$, $-a_2<f<+\infty$, equations \eqref{neq30} still hold. However in this case $f_3 > a_1$. Again $\eta \equiv 0$, but the trajectories close after passing the whole cross section $\alpha_3  = 0$.

In the intermediate case
\begin{equation}\label{neq31}
  f=-a_2, \qquad  h=a_3a_1
\end{equation}
two equilibrium points $(0,\pm 1,0)$ arise; they cut the cross section $\alpha_3  = 0$ into two segments on both of which the asymptotic motions in two directions take place. The set $\mfT_{a_3a_1,-a_2}$ is emphasized in Fig.~\ref{fig_02}, where the type of the manifolds $\mfT_{h_3,f}$ is shown while crossing the values \eqref{neq31}.

\begin{figure}[!ht]
\centering
\includegraphics[width=0.35\paperwidth,keepaspectratio]{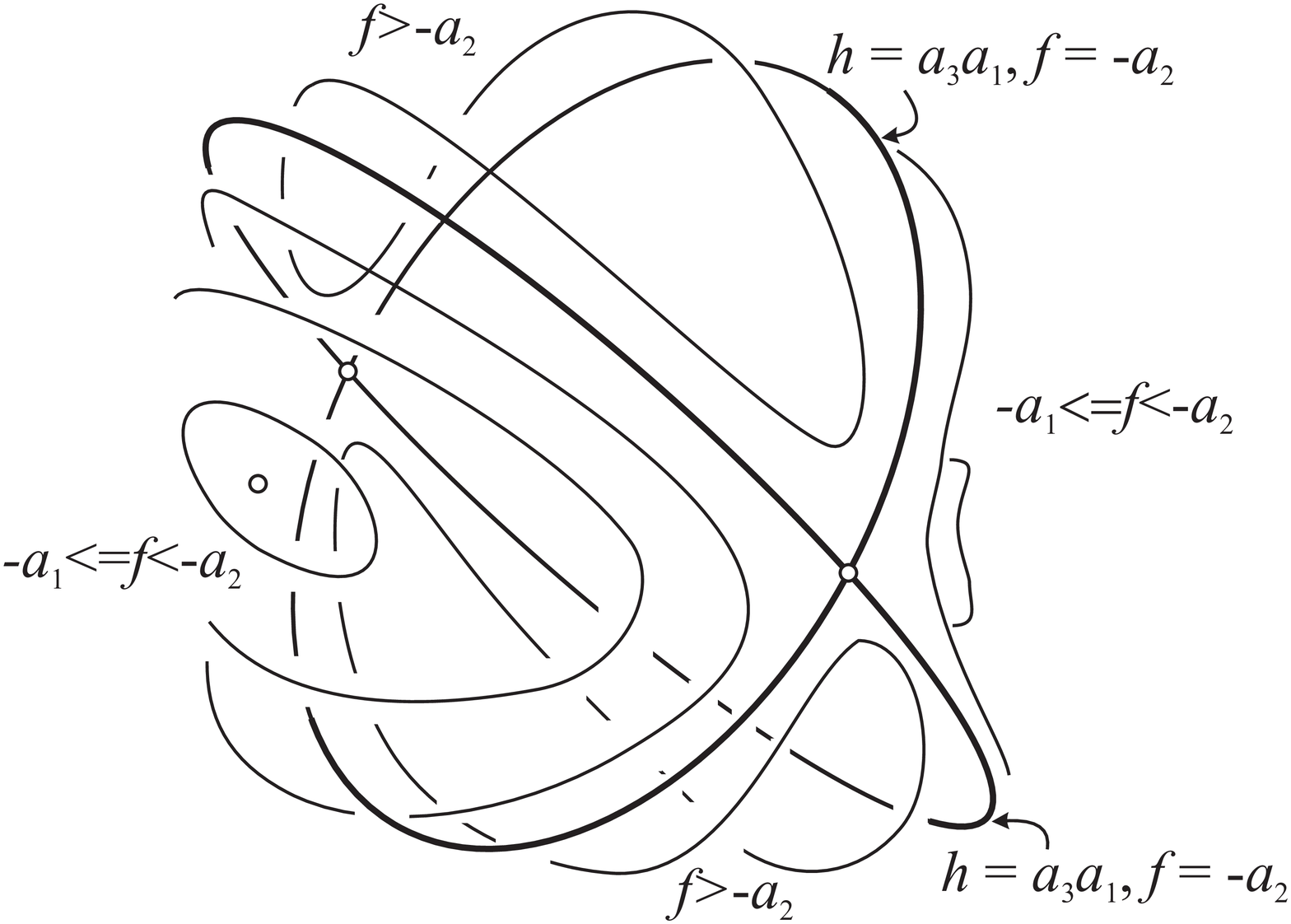}
\caption{}\label{fig_02}
\end{figure}

Suppose that the point $(h,f)$ is on the half-line \eqref{neq22}. Equations \eqref{neq16} take the form
\begin{equation}\label{neq32}
  \left(\frac{d\xi}{d\tau}\right)^2=\frac{(a_1-\lambda)(\lambda-f_1)}{\lambda}, \qquad \left(\frac{d\eta}{d\tau}\right)^2=\frac{(a_1-\mu)(f_1-\mu)}{\mu}.
\end{equation}
Consider the segment of the half-line \eqref{neq22} such that
\begin{equation}\label{neq33}
  -a_2<f<-a_3.
\end{equation}
It follows from \eqref{neq27} that $a_3<f_1<a_2$. The solutions of equations \eqref{neq32} lie in the ring
\begin{equation}\label{neq34}
  \mu \ls f_1.
\end{equation}
Among them there are two periodic oscillations along the segments of the cross section $\alpha_1=0$ cut out by the ring \eqref{neq34}. Let us show that all of them are limit cycles for other trajectories with \eqref{neq22} and \eqref{neq33}. To be definite suppose that such a trajectory at $\tau = 0$ is in the first octant and ${d\xi}/{d\tau}|_{\tau=0} > 0$. Since ${d\xi}/{d\tau}$ does not vanish, this sign in the first equation \eqref{neq32} can change only in the case when the trajectory crosses the parametric curve $\xi=m$ and enters the domain of another chart. But this can happen only at the moment
\begin{equation*}
  \tau=\int\limits_{\xi_0}^m\frac{\sqrt{\lambda}d\xi}{\sqrt{(a_1-\lambda)(\lambda-f_1}}=\int\limits_\lambda^{a_1}\frac{\lambda d\lambda}{(a_1-\lambda)\sqrt{(\lambda-f_1)(\lambda-a_2)(\lambda-a_3)}},
\end{equation*}
while the integral in the right-hand part obviously diverges. Therefore as $\tau \to +\infty$ the considered trajectory asymptotically approaches the cross section $\alpha_1=0$ (Fig.~\ref{fig_03}).

\begin{figure}[!ht]
\centering
\includegraphics[width=0.3\linewidth,keepaspectratio]{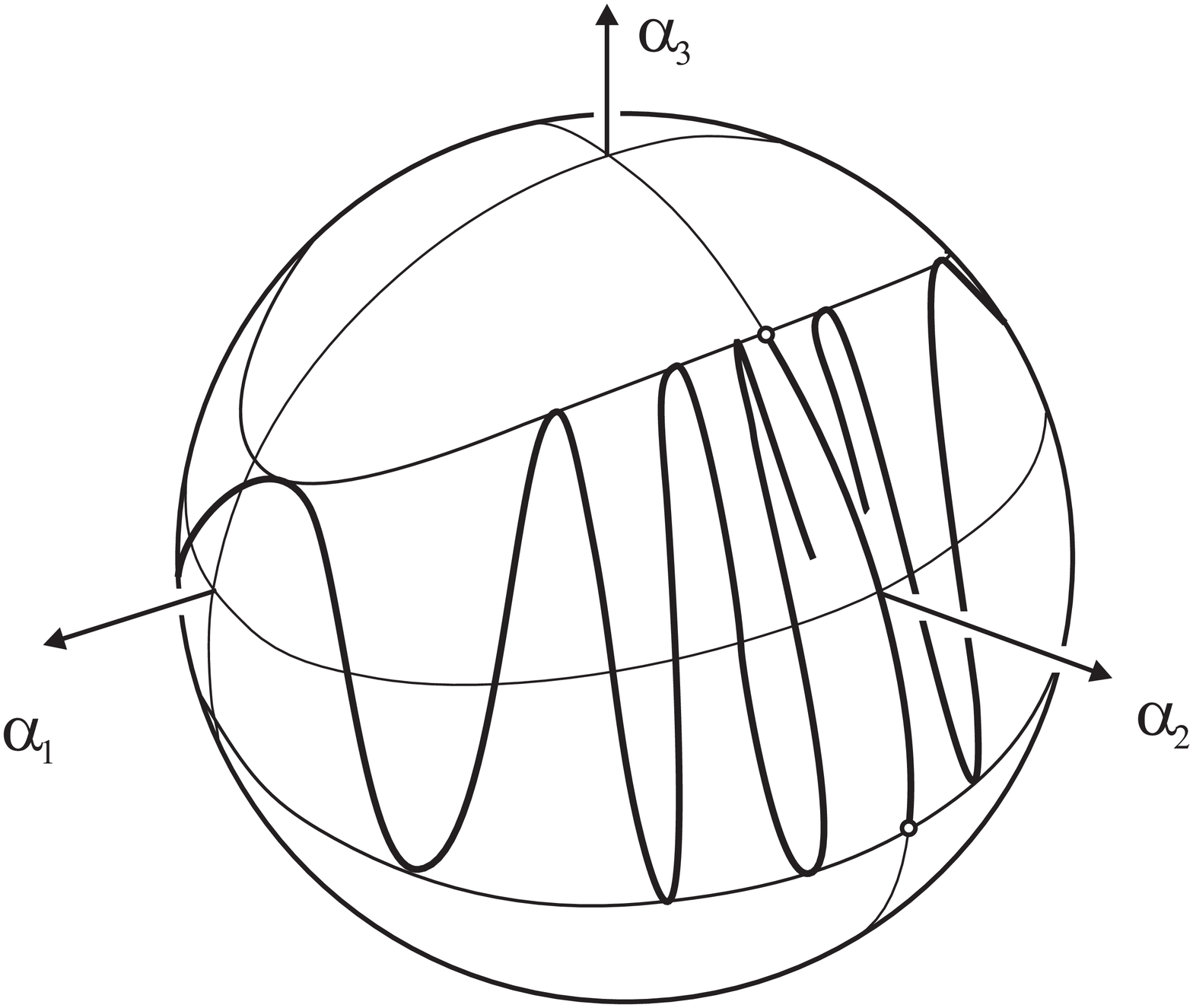}
\caption{}\label{fig_03}
\end{figure}

The family of the trajectories going in the same direction along the ring \eqref{neq34} gives, in the phase space, a torus with two limit cycles (see Fig.~\ref{fig_04}{\it  a}). The similar surface is filled by the trajectories going in the opposite direction (see Fig.~\ref{fig_04}{\it  b}). These two tori forming the set $\mfT_{h_1,f}$ intersect by two circles, namely, the limit cycles (see Fig.~\ref{fig_04}{\it  c}). Now it is easy to understand how this surface degenerates into $\mfT_{a_3a_1,-a_2}$ (see Fig.~\ref{fig_02}) as $f\searrow-a_2$; the limit cycles contract to the equilibrium points and the halves of the tori turn into half-circles which are the parts of the separatrix.

\begin{figure}[!ht]
\centering
\includegraphics[width=0.7\paperwidth,keepaspectratio]{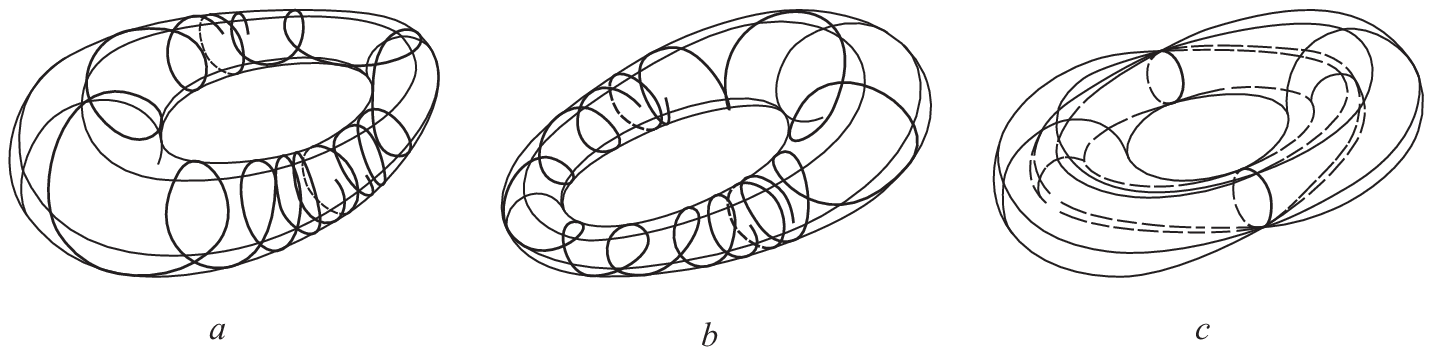}
\caption{}\label{fig_04}
\end{figure}

Let us consider the special point of the half-line \eqref{neq22} in which
\begin{equation}\label{neq35}
  f=-a_3, \qquad h=a_1a_2.
\end{equation}
Here $f_1=a_2$, and \eqref{neq32} yields the existence of two equilibrium points $(0,0,\pm 1)$. The oscillations along the cross section $\alpha_1=0$ about the points $(0,\pm 1,0)$ generate four segments of the separatrix similar to the picture shown in Fig.~\ref{fig_02}.

Motions along the cross section $\alpha_2=0$ are also the solutions of \eqref{neq32} in the case \eqref{neq35}. Any such motion approaches one equilibrium point as $t \to +\infty$ and another one as $t \to -\infty$.

Consider the reduced metric on the Poisson sphere; the square of the vector length is given by the quadratically homogeneous part of the function \eqref{neq8}. Let us denote by $P_1$, $P_2$, $P_3$, $P_4$ the umbilical points of the reduced metric with elliptic coordinates $\lambda = \mu = a_2$ lying respectively in the domains $\{\alpha_1>0, \alpha_3>0\}$, $\{\alpha_1<0, \alpha_3>0\}$, $\{\alpha_1<0, \alpha_3<0\}$, $\{\alpha_1>0, \alpha_3<0\}$.

\begin{propos}\label{prop2}
Each motion corresponding to \eqref{neq35}, except for the motions along the cross sections $\alpha_1=0$ and $\alpha_2 = 0$, crosses exactly one umbilical point. All trajectories containing $P_1$ or $P_2$ asymptotically approach $(0, 0, -1)$ as $t \to \pm \infty$. Similarly, all trajectories containing $P_3$ or $P_4$ asymptotically approach $(0, 0, 1)$.
\end{propos}

\begin{proof} Let $\nu(\tau)=(\xi(\tau), \eta(\tau))$ be a solution of equations \eqref{neq32}. To be definite, suppose that $(\xi(0), \eta(0))$ is in the first octant and ${d\xi}/{d\tau}(0)>0$, ${d\eta}/{d\tau}(0)<0$. From \eqref{neq32} we immediately obtain that $\lim\limits_{\tau\to+\infty}\xi(\tau)=m$. The coordinate $\eta(\tau)$ is decreasing monotonously and reaches the zero value at the time moment
\begin{equation*}
  \tau_0=\int\limits_0^{\eta(0)}\frac{\sqrt{\mu}d\eta}{\sqrt{(a_1-\mu)(a_2-\mu)}}=\int\limits_{a_3}^{\mu(0)}\frac{\mu d\mu}{(a_1-\mu)(a_2-\mu)\sqrt{\mu-a_3}}>0;
\end{equation*}
after that the trajectory enters the octant $\{\alpha_1>0$, $\alpha_2>0$, $\alpha_3<0\}$. In new coordinates, ${d\xi}/{d\tau}|_{\tau=\tau_0}>0$ and it follows from \eqref{neq32} that $\lim\limits_{\tau \to +\infty}\eta(\tau)=n$. Thus, $\lim\limits_{t \to +\infty}\nu(\tau(t)) = \lim\limits_{\tau \to +\infty} \nu(\tau) = (0,0,-1)$.

It is easily seen that $\lim\limits_{\tau \to -\infty}\xi(\tau)=0$ and $\lim\limits_{\tau \to -\infty}\eta(\tau)=n$. Consequently, $\lim\limits_{\tau \to -\infty}\nu(\tau)=P_1$. Nevertheless in the real time $t$ the point $P_1$ is not an equilibrium (the change \eqref{neq15} has singularities at the umbilical points). Therefore the trajectory will cross $P_1$ as $t$ decreases, and similar to the above case we get $\lim\limits_{t \to -\infty} \nu(t) = (0,0,-1)$ (Fig.~\ref{fig_05}). This proves the statement. \end{proof}

\begin{figure}[!ht]
\centering
\includegraphics[width=0.3\linewidth,keepaspectratio]{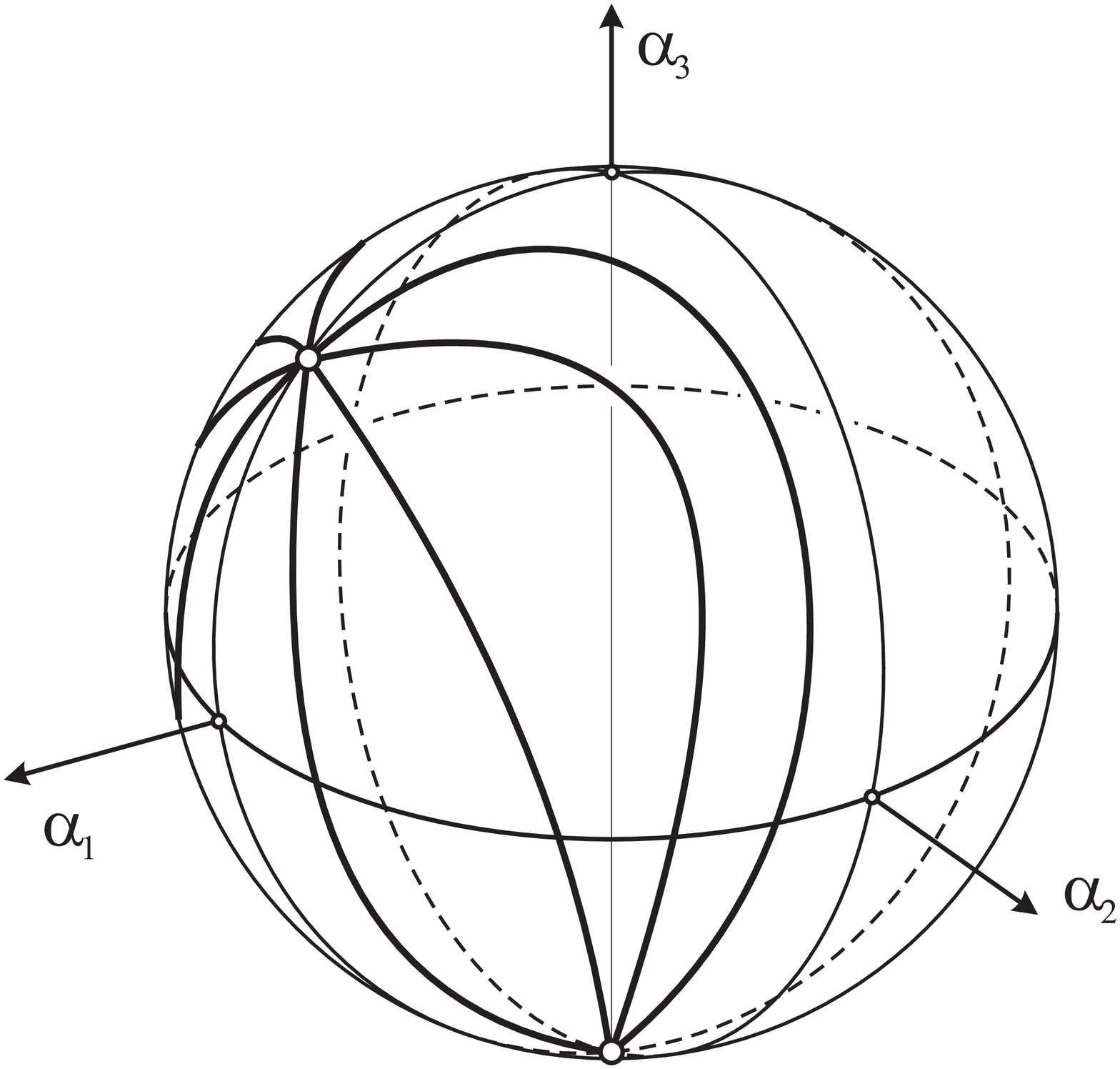}
\caption{}\label{fig_05}
\end{figure}

To construct the integral set $\mfT_{a_1a_2,-a_3}$ we consider first the trajectories in the hemisphere $\alpha_1 \gs 0$. The scheme of construction is shown in Fig.~\ref{fig_06}~--~\ref{fig_09}. The segments {\it  1}, {\it  2} and {\it  3}, {\it  4} in the phase space give two circles composed of the admissible velocities of motions from the points $P_1$ and $P_4$. Pairwise unions of the segments {\it  11}, {\it  10} and {\it  12}, {\it  9} give the same trajectory coded by {\it  21}; it is the asymptotical motion along the cross section $\alpha_2=0$ from the equilibrium point $I$ with coordinates $(0,0,1)$ to the equilibrium point $II$ with coordinates $(0,0,-1)$. The unions of {\it  16}, {\it  13} and {\it  15}, {\it  14} give a similar motion from $II$ to $I$ coded as {\it  29}. Now we glue the regions in Fig.~\ref{fig_06}{\it a} and Fig.~\ref{fig_06}{\it b} along the curves {\it  1} and {\it  2}, and after that we cut the result along {\it  10}, {\it  11}, {\it  13} and {\it  16}. We obtain two squares in Fig.~\ref{fig_07}{\it a} and Fig.~\ref{fig_07}{\it b}. Analogously, from the regions in Fig.~\ref{fig_06}{\it c} and Fig.~\ref{fig_06}{\it d} we get the squares in Fig.~\ref{fig_07}{\it c} and Fig.~\ref{fig_07}{\it d}. Identifying the corresponding sides of the squares in Fig.~\ref{fig_07}{\it b} and  Fig.~\ref{fig_07}{\it c}, we get the torus shown in Fig.~\ref{fig_08}{\it a}. The equilibrium points $I$ and $II$ lie in its directrix consisting of the trajectories {\it  21} and {\it  29}. To see the character of all other trajectories is an easy task. In the same way, the squares in Fig.~\ref{fig_07}{\it a} and Fig.~\ref{fig_07}{\it   d} generate the torus in Fig.~\ref{fig_08}{\it b}. The tori of Fig.~\ref{fig_08}{\it a} and Fig.~\ref{fig_08}{\it  b} intersect along the common directrix and form a surface shown in Fig.~\ref{fig_09}.
Another such surface is formed in the phase space by the trajectories lying in the hemisphere ${\alpha_1\ls 0}$. These two surfaces intersecting at the trajectories {\it 5 -- 8} (the motions along the cross section $\alpha_1 = 0$) give the integral set $\mfT_{a_1a_2,-a_3}$.

\begin{figure}[ht]
\begin{center}
\begin{minipage}[h]{0.4\linewidth}
\centering
\includegraphics[width=1\linewidth,keepaspectratio]{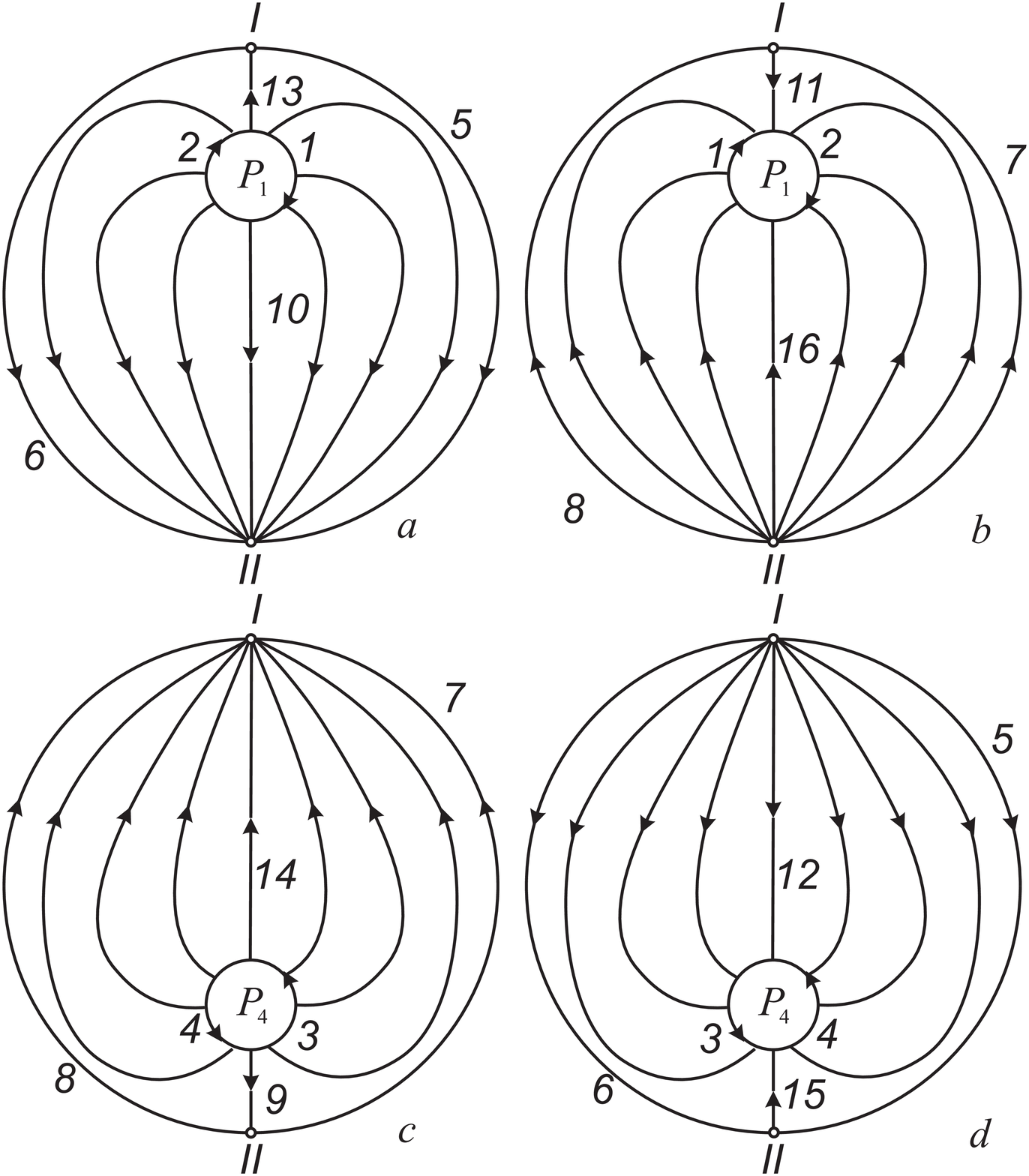}
\caption{}\label{fig_06}
\end{minipage}
\hspace{5mm}
\begin{minipage}[h]{0.4\linewidth}
\centering
\includegraphics[width=1\linewidth,keepaspectratio]{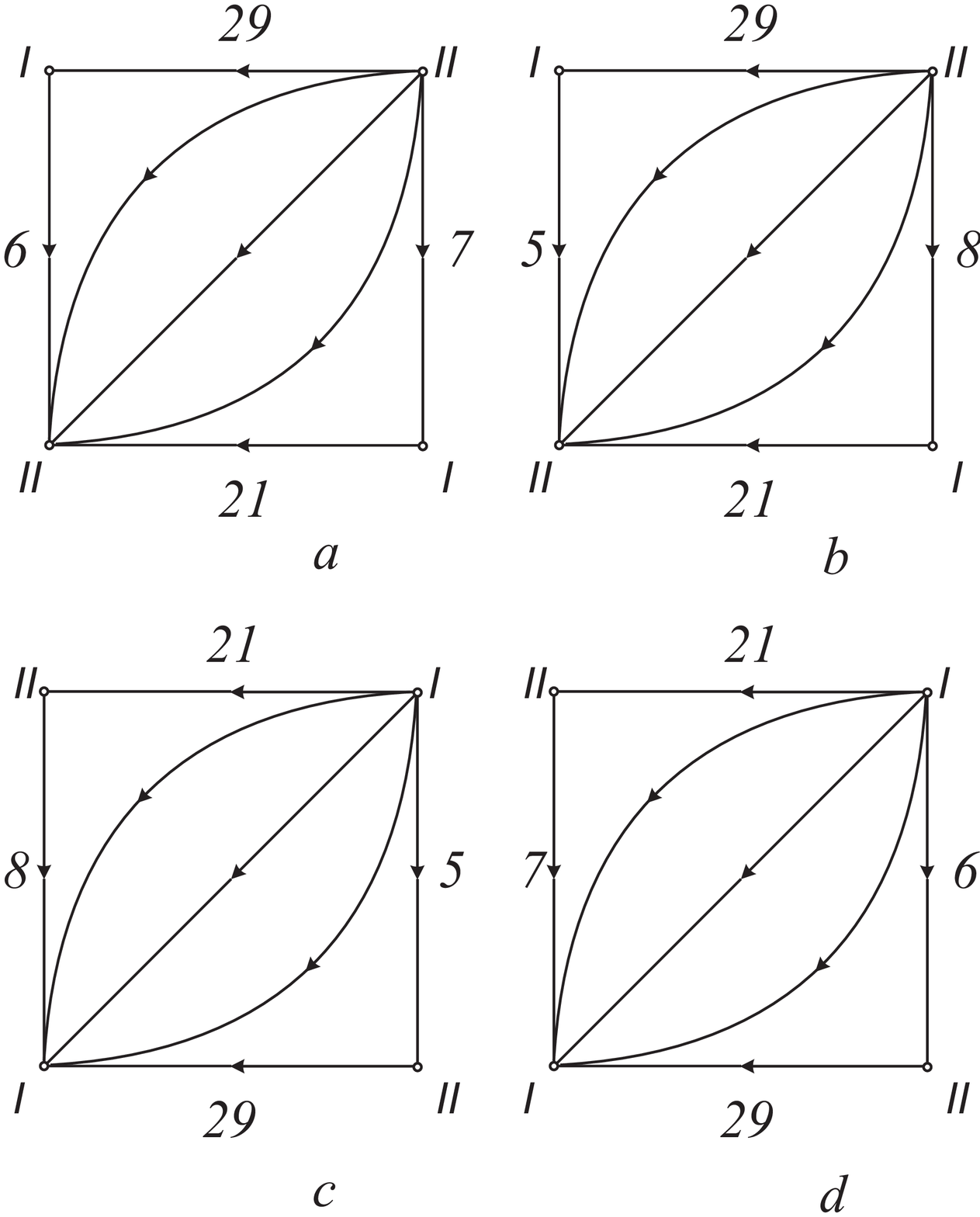}
\caption{}\label{fig_07}
\end{minipage}
\end{center}
\end{figure}

\begin{figure}[!ht]
\centering
\includegraphics[width=0.7\paperwidth,keepaspectratio]{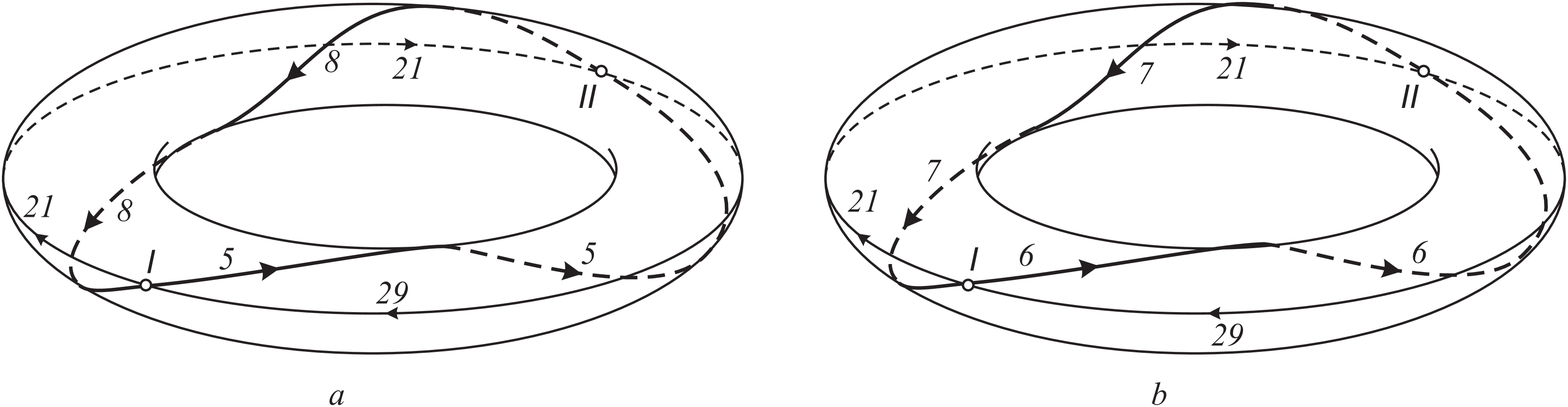}
\caption{}\label{fig_08}
\end{figure}

\begin{figure}[!ht]
\centering
\includegraphics[width=0.35\paperwidth,keepaspectratio]{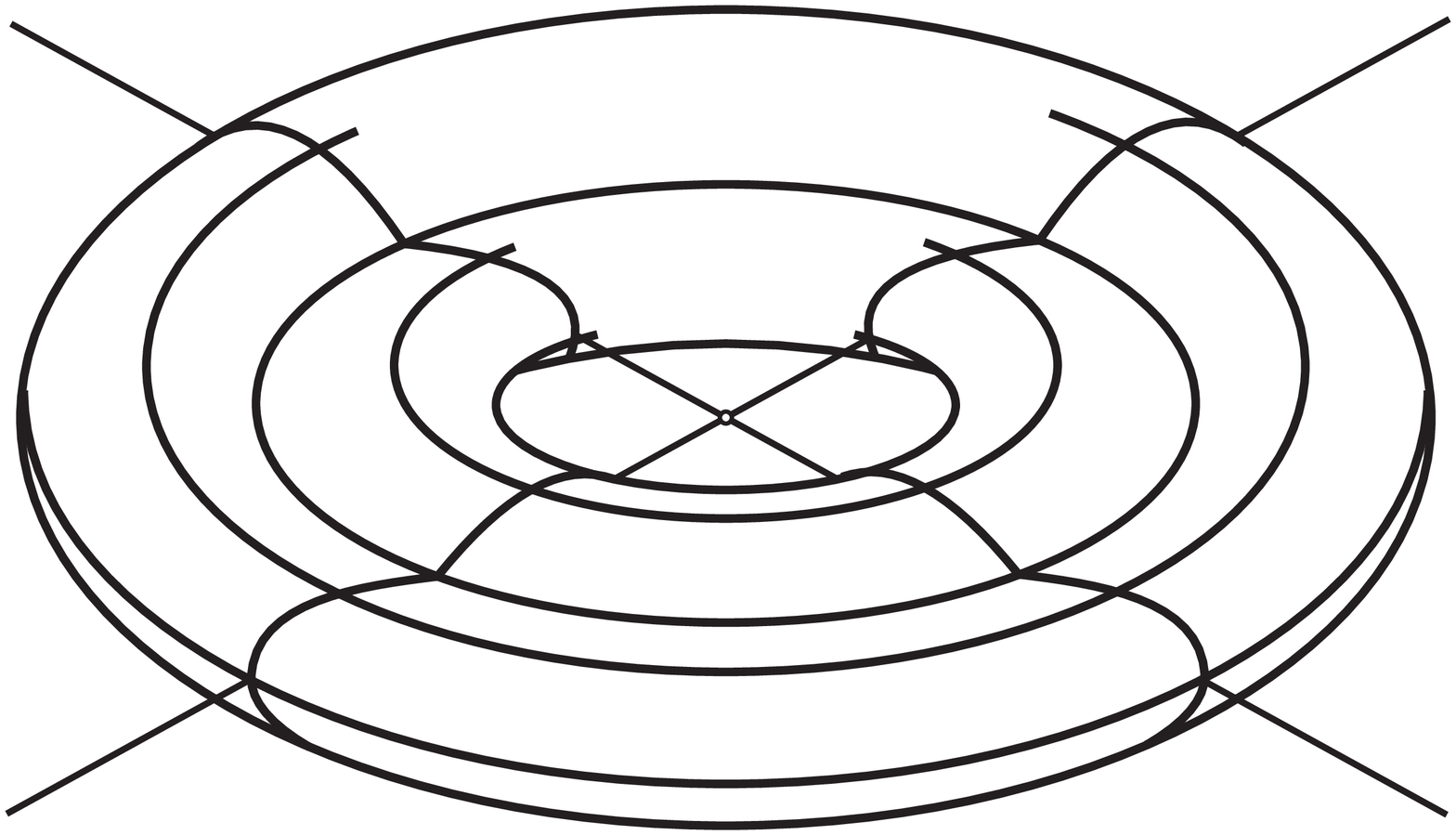}
\caption{}\label{fig_09}
\end{figure}

Let again \eqref{neq22} hold and $-a_3<f<a_1-(a_2+a_3)$. So $a_2<f_1<a_3$. The region of possible motions $\tilde{M}_{h_1,f}$ is a ring
\begin{equation}\label{neq36}
  f_1 \ls \lambda \ls a_1.
\end{equation}

Equations \eqref{neq32} admit two periodic motions which are rotations along the cross section $\alpha_1 = 0$ in two directions. The cross section $\alpha_1 = 0$ splits \eqref{neq36} into two strips and any trajectory starting in one of the strips remains it it forever. The elementary analysis of \eqref{neq32} shows that any such trajectory touches the outer boundary of the corresponding strip exactly once and, as $t \to \pm \infty$, asymptotically approaches the cross section $\alpha_1 = 0$ (see Fig.~\ref{fig_10}). Each strip with a fixed motion direction on it gives a torus in the phase space. One of its directrix is the corresponding periodic motion. Along this directrix the torus intersects with another torus given by another strip with the same motion direction. Therefore $\mfT_{h_1,f}$ consists of two connected components each of which is homeomorphic to the surface shown in Fig.~\ref{fig_09}.

\begin{figure}[!ht]
\centering
\includegraphics[width=0.3\linewidth,keepaspectratio]{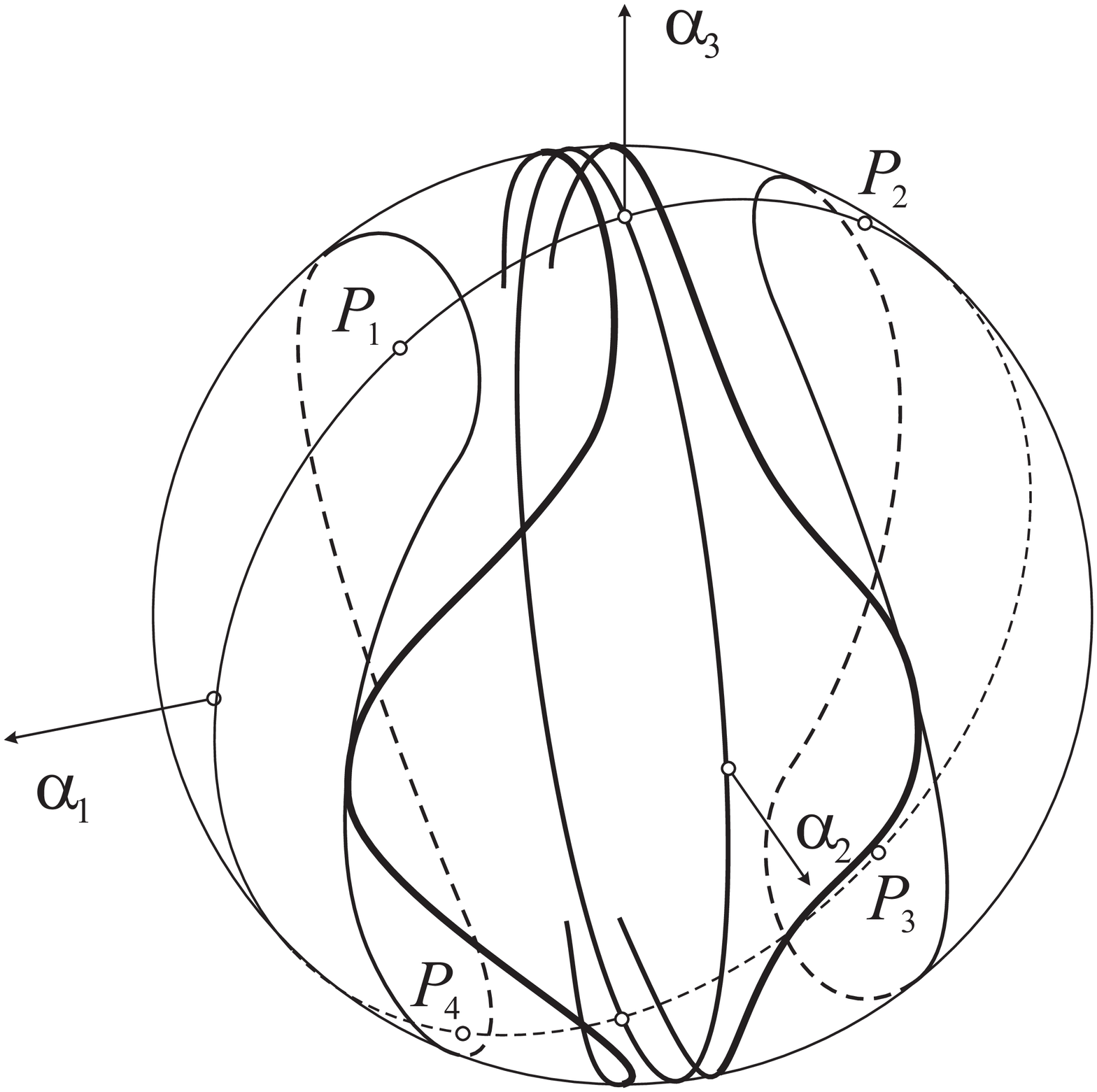}
\caption{}\label{fig_10}
\end{figure}

In the case
\begin{equation}\label{neq37}
  f \gs a_1-(a_2+a_3)
\end{equation}
the ring \eqref{neq36} degenerates into the cross section $\alpha_1 = 0$. The manifold $\mfT_{h_1,f}$ consists of two non-intersecting circles corresponding to the motions along this cross section in two directions.

Now we consider the points of the half-line \eqref{neq23}. Substituting the values from \eqref{neq23} into \eqref{neq16}, we get
\begin{equation}\label{neq38}
  \left(\frac{d\xi}{d\tau}\right)^2=\frac{(\lambda-a_2)(f_2-\lambda)}{\lambda}, \qquad \left(\frac{d\eta}{d\tau}\right)^2=\frac{(a_2-\mu)(f_2-\mu)}{\mu}.
\end{equation}

Let $-a_1 < f \ls a_2-(a_3+a_1)$, then according to notation \eqref{neq27}, $a_3 < f_2 \ls a_2$, so it follows from the first equation \eqref{neq38} that $\xi \equiv 0$. The solutions in this case are oscillations along the cross section $\alpha_2 = 0$ about the points $(\pm 1,0,0)$ in the limits $\mu \ls f_2$. If
\begin{equation}\label{neq39}
  f_2=a_2,
\end{equation}
then one of the oscillations reaches the umbilical points $P_1$ and $P_4$, while the other one reaches $P_2$ and $P_3$. The manifold $\mfT_{h_2,f}$ consists of two non-intersecting circles.

In the case when
\begin{equation}\label{neq40}
  a_2-(a_3+a_1)<f<-a_3
\end{equation}
we get $a_2<f_2<a_1$, so the region of possible motions is the union of two disks
\begin{equation}\label{neq41}
  \lambda \ls f_2.
\end{equation}
Let us investigate the motions in the domain $\{\alpha_1>0\}$.

\begin{propos}\label{prop3} Along the segment of the cross section $\alpha_2=0$ we have a periodic oscillation about the point $(1,0,0)$ in the limits \eqref{neq41}. Any other solution cross, by turns, the umbilical points $P_1$ and $P_4$. While moving from $P_4$ to $P_1$, the same as from $P_1$ to $P_4$, the trajectory exactly once touches the boundary of the disk \eqref{neq41}.
\end{propos}

\begin{proof} The first statement is obvious. Consider any other solution
\begin{equation}\label{neq42}
  \nu(\tau)=(\lambda(\tau),\mu(\tau)), \quad  -\infty < \tau <+\infty,
\end{equation}
of equations \eqref{neq38}. Any segment \eqref{neq42} of a real solution parameterized with the ``amended time'' $\tau$ will be called a pseudo-trajectory.

To be definite suppose that $\nu(0)=(\lambda_0, \mu_0)$ lies in the first octant and ${d\lambda}/{d\tau}(0)>0$, ${d\mu}/{d\tau}(0)>0$.

On the set $[a_3,a_2)\times[a_3,a_2)\cup(a_2,f_2]\times(a_2,f_2]$, let us define the function of two variables
\begin{equation}\label{neq43}
  C(x,y)=\int\limits_x^y\frac{zdz}{(z-a_2)\sqrt{(a_1-z)(f_2-z)(z-a_3)}}
\end{equation}
and put
\begin{equation}\label{neq44}
  \tau_0=C(\lambda_0,f_2)>0, \qquad \tau_1=C(a_3, \mu_0)<0.
\end{equation}
It follows from \eqref{neq14}, \eqref{neq38} that the point $\nu(\tau_1)=(\lambda(\tau_1), a_3)$ belongs to the curve $\alpha_3=0$, and the point $\nu(\tau_0)=(f_2, \mu(\tau_0))$ lies on the boundary \eqref{neq41}. Moreover for $\tau \ls \tau_0$ the dependency $\lambda(\tau)$ is defined from the equation
\begin{equation}\label{neq45}
  \tau_0-\tau=C(\lambda(\tau), f_2),
\end{equation}
and for $\tau \gs \tau_0$,
\begin{equation}\label{neq46}
  \tau-\tau_0=C(\lambda(\tau), f_2).
\end{equation}
Analogously, for $\tau \ls \tau_1$ we have
\begin{equation}\label{neq47}
  \tau_1-\tau=C(\mu(\tau), a_3),
\end{equation}
and for $\tau \gs \tau_1$,
\begin{equation}\label{neq48}
  \tau-\tau_1=C(\mu(\tau), a_3).
\end{equation}

On the mentioned segments of $\tau$ the values $\lambda$ and $\mu$ change monotonously, therefore the tangency point of a pseudo-trajectory with the parametric curve $\lambda=f_2$ is unique, namely, $\tau = \tau_0$. Passing to the limits as $\tau \to \pm \infty$ in \eqref{neq45}~--~\eqref{neq48} and taking into account the structure of the integral in \eqref{neq43} we get
\begin{equation}\label{neq49}
  \lim\limits_{\tau\to\pm\infty}\lambda(\tau)=\lim\limits_{\tau\to\pm\infty}\mu(\tau)=a_2.
\end{equation}
Thus $\lim\limits_{\tau\to -\infty}\nu(\tau)=P_4$ and $\lim\limits_{\tau\to +\infty}\nu(\tau)=P_1$. All other cases of the position of $\nu(0)$ and the initial velocity directions are considered similarly. Since $P_1$ and $P_4$ are not equilibrium points in the real time $t$, the real trajectory crosses these points by turn. The proposition is proved.\end{proof}

Let us find out the global character of the trajectories corresponding to the values \eqref{neq23} in the interval \eqref{neq40}. To each pseudo-trajectory \eqref{neq42} we assign two numbers
\begin{equation}\label{neq50}
  \varphi^+(\nu)=\lim\limits_{\tau\to +\infty}\frac{\lambda(\tau)-a_2}{a_2-\mu(\tau)}, \qquad \varphi^-(\nu)=\lim\limits_{\tau\to -\infty}\frac{\lambda(\tau)-a_2}{a_2-\mu(\tau)}.
\end{equation}

\begin{propos}\label{prop4} The limits \eqref{neq50} exist. The product $\theta=\varphi^+(\nu)\varphi^-(\nu)$ depends only on the constant $f$. If $\theta=1$, then all trajectories are closed. If $\theta \neq 1$, then we have a unique periodic solution, namely, the motion along the segment of the cross section $\alpha_2=0$; all other trajectories asymptotically approach it as $t\to\pm\infty$ $($see Fig.~$\ref{fig_11})$.
\end{propos}

\begin{figure}[!ht]
\centering
\includegraphics[width=0.6\paperwidth,keepaspectratio]{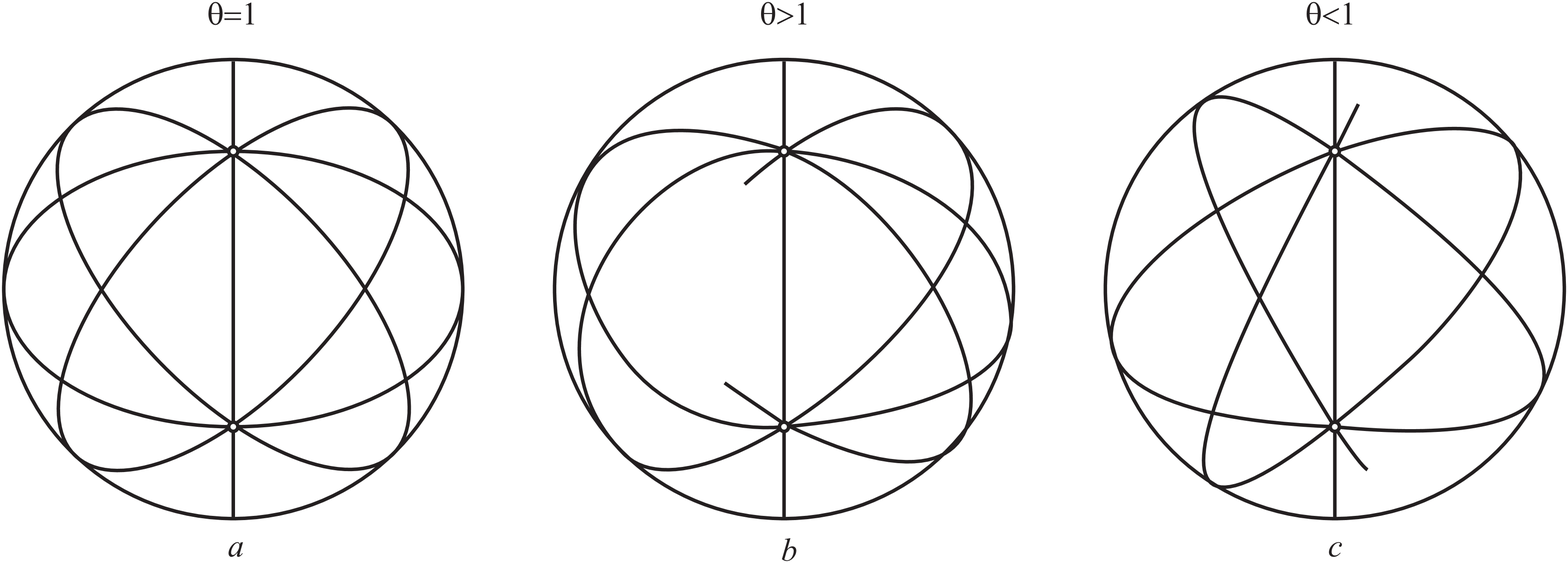}
\caption{}\label{fig_11}
\end{figure}

\begin{proof} Let us present the calculations for a pseudo-trajectory \eqref{neq42} such that
$$
\nu(0)=(f_2,\mu_0), \qquad \ds{\frac{d\mu}{d\tau}(0)>0}.
$$
Equations \eqref{neq45} -- \eqref{neq48} hold in view of the fact that in notation \eqref{neq44} we have $\tau_0=0$.

We rewrite the integral \eqref{neq43} in the form
\begin{equation}\label{neq51}
  C(x,y)=s\ln\frac{y-a_2}{x-a_2}+\int\limits_{x-a_2}^{y-a_2}\Psi(\zeta)d\zeta,
\end{equation}
where $\ds{s=\frac{a_2}{\sqrt{a_1^*f_2^*a_3^*}}}$, $a_1^*=a_1-a_2$, $f_2^*=f_2-a_2$, $a_3^*=a_2-a_3$, and
\begin{equation}\label{neq52}
  \Psi(\zeta)=\frac{\zeta+a_2}{\zeta\sqrt{(a_1^*-\zeta)(f_2^*-\zeta)(\zeta+a_3^*)}}-\frac{s}{\zeta}
\end{equation}
is continuous in the neighborhood of the value $\zeta=0$.

Let $\tau>0$. Equalities \eqref{neq46}, \eqref{neq48} together with \eqref{neq44}, \eqref{neq51} give
\begin{equation*}
\begin{array}{c}
  \ds{\tau=s\ln\frac{f_2-a_2}{\lambda(\tau)-a_2}+\int\limits^{f_2-a_2}_{\lambda(\tau)-a_2}\Psi(\zeta)d\zeta,}\quad \ds{\tau=s\ln\frac{a_2-\mu_0}{a_2-\mu(\tau)}+\int\limits^{\mu_0-a_2}_{\mu(\tau)-a_2}\Psi(\zeta)d\zeta.}
\end{array}
\end{equation*}
Subtract the first value from the second one:
\begin{equation*}
  s\ln\left(\frac{a_2-\mu_0}{f_2-a_2}\frac{\lambda(\tau)-a_2}{a_2-\mu(\tau)}\right) = \int\limits^{f_2-a_2}_{\lambda(\tau)-a_2}\Psi(\zeta)d\zeta+\int\limits_{\mu_0-a_2}^{\mu(\tau)-a_2}\Psi(\zeta)d\zeta.
\end{equation*}
In view of \eqref{neq49} let $\tau$ tend to $+\infty$. Then we get
\begin{equation*}
  \varphi^+(\nu)=\frac{f_2-a_2}{a_2-\mu_0}\exp\left[\frac{1}{s}\int\limits^{f_2-a_2}_{\mu_0-a_2}\Psi(\zeta)d\zeta\right].
\end{equation*}

Thus in the case $\tau < \tau_1$,  \eqref{neq45}, \eqref{neq47}, and \eqref{neq51} yield
\begin{equation*}
\begin{array}{c}
  \ds{-\tau=s\ln\frac{f_2-a_2}{\lambda(\tau)-a_2}+\int\limits_{\lambda(\tau)-a_2}^{f_2-a_2}\Psi(\zeta)d\zeta,}\\
  \ds{-\tau=s\ln\frac{(a_2-a_3)^2}{(a_2-\mu_0)(a_2-\mu(\tau))}+\int\limits_{\mu(\tau)-a_2}^{a_3-a_2}\Psi(\zeta)d\zeta+\int\limits_{\mu_0-a_2}^{a_3-a_2}\Psi(\zeta)d\zeta}.
\end{array}
\end{equation*}
Excluding $\tau$ and passing to the limit as $\tau\to -\infty$, we find
\begin{equation*}
  \ds{\varphi^-(\nu)=\frac{(a_2-\mu_0)(f_2-a_2)}{(a_2-a_3)^2}\exp\left[\frac{1}{s}\left(\int\limits_{a_3-a_2}^{f_2-a_2} \Psi(\zeta)d\zeta+\int\limits_{a_3-a_2}^{\mu_0-a_2}\Psi(\zeta)d\zeta \right)\right]}.
\end{equation*}
Consequently the value
\begin{equation}\label{neq53}
  \ds{\theta=\varphi^+(\nu)\varphi^-(\nu)=\left(\frac{f_2-a_2}{a_2-a_3}\right)^2\exp\left[\frac{2}{s}\int\limits_{a_3-a_2}^{f_2-a_2}\Psi(\zeta)d\zeta\right]}
\end{equation}
does not depend on $\mu_0$ and is, therefore, the function only of the parameter $f$.

Suppose pseudo-trajectories $\nu_1(\tau)$ and $\nu_2(\tau)$ are chosen in such a way that, geometrically, $\nu_2$ can be the extension of $\nu_1$ in real time. Let us find out when it really happens. Consider, for example, a pseudo-trajectory $\nu_1$ going from $P_4$ to $P_1$ in the domain $\alpha_2>0$. Then $\nu_2$ must go from $P_1$ to $P_4$ in the domain $\alpha_2 < 0$. Obviously, $\nu_2$ is the extension of $\nu_1$ if and only if $\nu_1$ and $\nu_2$ have at the point $P_1$ a common tangent line. Substituting \eqref{neq23} and \eqref{neq40} into \eqref{neq8} and \eqref{neq9}, we get that at the umbilical points the admissible velocities form the circle
\begin{equation*}
  (\alpha_3\dot{\alpha}_1-\alpha_1\dot{\alpha}_3)^2+\dot{\alpha}_2^2=a_2^2(f_2-a_2).
\end{equation*}
It means that $\nu_1$ goes further into $\nu_2$ if and only if the following limits coincide:
$$
\ds{\psi^+(\nu_1)=\lim\limits_{\tau\to +\infty}\left(\alpha_3\frac{d\alpha_1}{dt}-\alpha_1\frac{d\alpha_3}{dt}\right)}
$$
calculated along $\nu_1(\tau(t))$ and
$$
\ds{\psi^-(\nu_2)=\lim\limits_{\tau\to -\infty}\left(\alpha_3\frac{d\alpha_1}{dt}-\alpha_1\frac{d\alpha_3}{dt}\right)}
$$
calculated along $\nu_2(\tau(t))$. Using \eqref{neq11} and \eqref{neq38}, we find
\begin{equation}\label{neq54}
\begin{array}{c}
  \ds{\psi^+(\nu_1)=a_2\sqrt{f_2-a_2}\;\frac{\varphi^+(\nu_1)-1}{\varphi^+(\nu_1)+1}, \qquad \psi^-(\nu_1)=a_2\sqrt{f_2-a_2}\;\frac{\varphi^-(\nu_1)-1}{\varphi^-(\nu_1)+1},}\\
  \ds{\psi^+(\nu_2)=a_2\sqrt{f_2-a_2}\;\frac{1-\varphi^+(\nu_2)}{1+\varphi^+(\nu_2)}, \qquad \psi^-(\nu_2)=a_2\sqrt{f_2-a_2}\;\frac{1-\varphi^-(\nu_2)}{1+\varphi^-(\nu_2)}.}
\end{array}
\end{equation}
Hence $\psi^-(\nu_2)=\psi^+(\nu_1)$ if and only if
\begin{equation}\label{neq55}
  \varphi^-(\nu_2)=1/\varphi^+(\nu_1).
\end{equation}

Any real trajectory $\nu(t)$ is a sequence of pseudo-trajectories
\begin{equation*}
  \ldots,\nu_{-n},\ldots\nu_{-1},\nu_0,\nu_1,\ldots, \nu_n,\ldots,
\end{equation*}
and according to \eqref{neq55}
\begin{equation}\label{neq56}
  \varphi^-(\nu_n)=1/\varphi^+(\nu_{n-1}), \qquad n=0,\pm1, \pm2, \ldots.
\end{equation}
From \eqref{neq53} and \eqref{neq55} we get
\begin{equation}\label{neq57}
  \ds{\varphi^-(\nu_n)=\frac{1}{\theta^n}\varphi^-(\nu_0),\qquad \varphi^+(\nu_n)=\theta^n\varphi^+(\nu_0).}
\end{equation}
Suppose that the considered trajectory is closed. Then for some $n>0$ we obtain $\varphi^+(\nu_n)=1/\varphi^-(\nu_0)$, and equalities \eqref{neq53} and \eqref{neq57} yield $\theta=1$. Note that in this case any two successive pseudo-trajectories already form a closed solution (see Fig.~\ref{fig_11}{\it  a}).

When $\theta>1$, relations \eqref{neq57} give that
$$
\lim\limits_{n\to +\infty}\varphi^+(\nu_n)=+\infty, \qquad \lim\limits_{n\to -\infty}\varphi^+(\nu_n)=0,
$$
therefore \eqref{neq54} yields
$$
\lim\limits_{n\to\pm\infty}\left|\psi^+(\nu_n)\right|=a_2\sqrt{f_2-a_2}.
$$
Consequently, the trajectory asymptotically approaches the periodic one as $t \to \pm \infty$. The similar phenomenon takes place when $\theta<1$. The difference here is that if $\theta>1$, then the sum of the inner angles at the ``basement'' of each pseudo-trajectory is greater than $\pi$, and if $\theta<1$ this sum is smaller than $\pi$ (see Fig.~\ref{fig_11}{\it  b} and Fig.~\ref{fig_11}{\it c}). The proposition is proved.\end{proof}

From the topological point of view, the distribution of admissible velocities in the disk \eqref{neq41} is the same as the distribution of admissible velocities in the domain $\{\alpha_1\gs 0\}$ of the case \eqref{neq35}. Therefore the trajectories in the phase space corresponding to the motions with $\alpha_1>0$ fill a surface homeomorphic to that shown in Fig.~\ref{fig_09}. However now the middle line of this surface is a periodic solution and all other trajectories behave in such a way that the first recurrence map induced by the phase flow on the eight-curve is either identity $(\theta = 1)$, or has a unique (hyperbolic) fixed point (see Fig.~\ref{fig_12}).

The surface of the same topology is filled by the trajectories covering the domain $\{\alpha_1< 0\}$. The manifold $\mfT_{h_2, f}$ in the case \eqref{neq23}, \eqref{neq40} has two components each of which is homeomorphic to the direct product of the eight-curve and a circle.

\begin{figure}[!ht]
\centering
\includegraphics[width=0.6\linewidth,keepaspectratio]{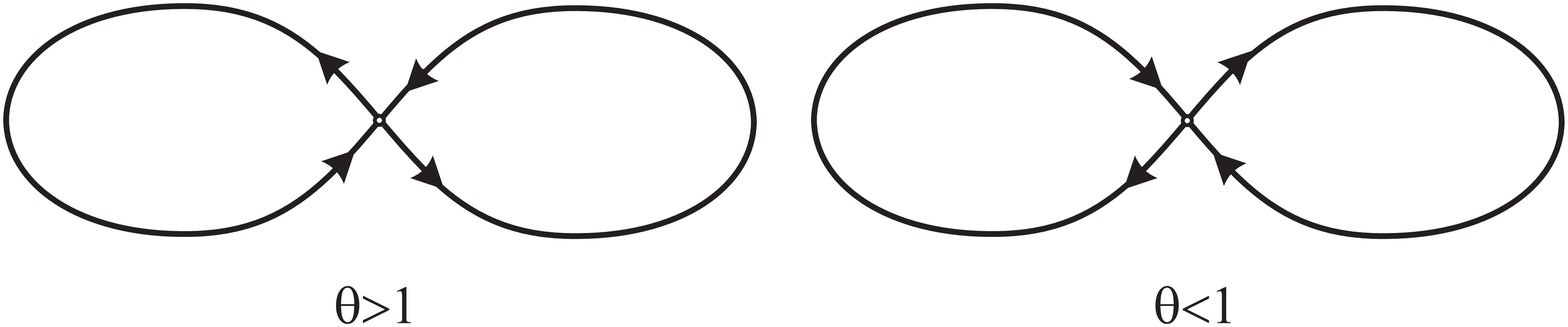}
\caption{}\label{fig_12}
\end{figure}

Let us consider the values \eqref{neq23} when
\begin{equation}\label{neq58}
  -a_3<f<+\infty, \qquad a_1<f_2<+\infty.
\end{equation}
From \eqref{neq26}, we obtain that the region of possible motions is the whole Poisson sphere.

\begin{propos}\label{prop5} Among the motions corresponding to the values \eqref{neq23} in the interval \eqref{neq58}, only two are periodic and these are the motions along the cross section $\alpha_2 = 0$ in two directions. All other trajectories cross by turn a pair of the opposite umbilical points $($see Fig.~$\ref{fig_13})$. The trajectories starting from one of the umbilical points at $t=0$ simultaneously join at the opposite umbilical point.
\end{propos}

\begin{figure}[ht]
\begin{center}
\begin{minipage}[h]{0.4\linewidth}
\centering
\includegraphics[width=0.75\linewidth,keepaspectratio]{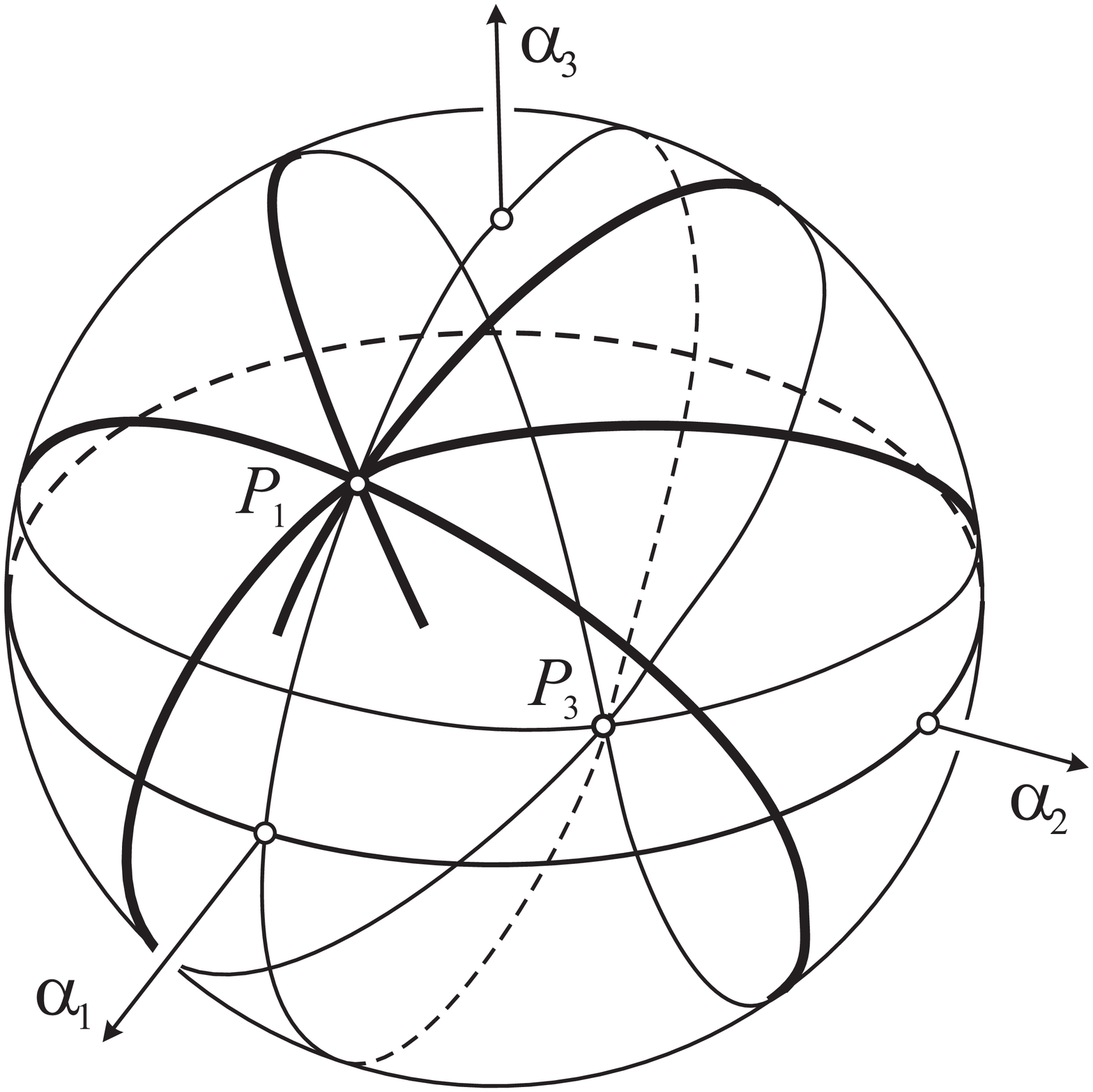}
\caption{}\label{fig_13}
\end{minipage}
\hspace{5mm}
\begin{minipage}[h]{0.4\linewidth}
\centering
\includegraphics[width=0.75\linewidth,keepaspectratio]{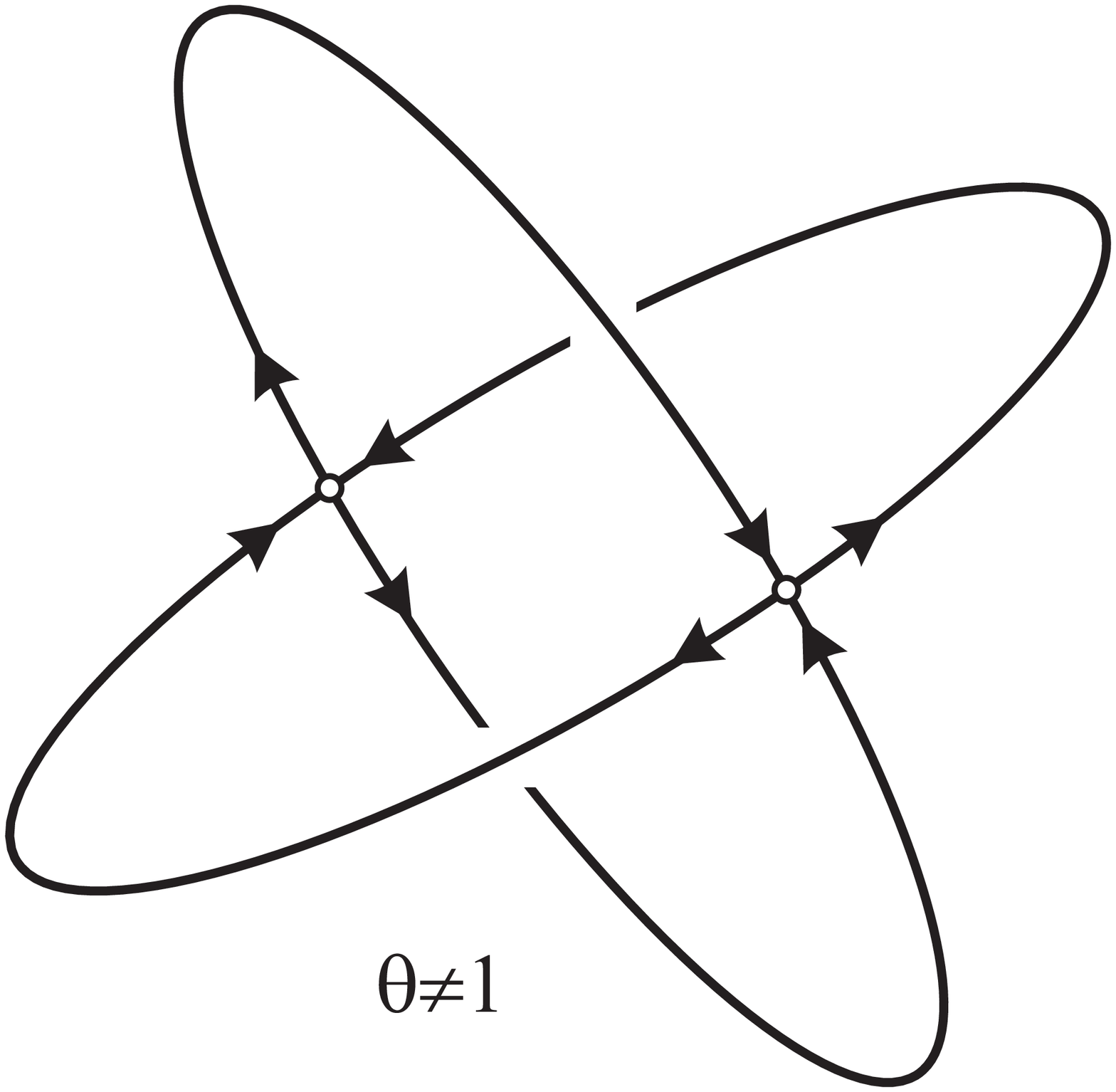}
\caption{}\label{fig_14}
\end{minipage}
\end{center}
\end{figure}

This case is an analog of a similar case in the problem of the inertial motion of a rigid body \cite{bib04}, and the proof of Proposition~\ref{prop5} is obtained from \eqref{neq38}, \eqref{neq58} by the same way as it is done in \cite{bib04}. This coincidence is easily explained. The constants $h$ and $f$ of the quadratic integrals are big enough to speak of the ``fast'' rotations of the body, therefore the force function does not influence the qualitative picture of motion.

To find out if the trajectories are closed or not, as before, we have to estimate the value (in the same notation)
\begin{equation*}
  \ds{\theta=\left(\frac{a_1-a_2}{a_2-a_3}\right)^2\exp\left[\frac{2}{s} \int\limits_{a_3-a_2}^{a_1-a_2}\Psi(\zeta)d\zeta\right]}.
\end{equation*}
Only now in the function \eqref{neq52} we have $f_2^*>a_1^*$. If $\theta=1$, then all trajectories are closed, and if $\theta \neq 1$, then each trajectory asymptotically approaches one of the periodic motions mentioned in Proposition~\ref{prop5} as $t \to +\infty$ and another one as $t \to -\infty$.

The motions crossing $P_1$ and $P_3$ form, in the phase space, a two-dimensional torus $T_1$, and the motions crossing $P_2$ and $P_4$ give a two-dimensional torus $T_2$. The tori $T_1$ and $T_2$ intersect by two circles corresponding to the motions along the cross section $\alpha_2 = 0$. Thus, the set $\mfT_{h_2, f}$ for the values \eqref{neq58} is homeomorphic to the set shown in Fig.~\ref{fig_04}{\it c}. On the cross section of $\mfT_{h_2, f}$ transversal to the common directrices of the tori $T_1$ and $T_2$ the first recurrence map is defined. It is the identity for $\theta=1$ and has two hyperbolic fixed points for при $\theta \neq 1$ (see Fig.~\ref{fig_14}).

It remains to investigate the parabolic segment \eqref{neq25}, which according to \eqref{neq27} can be written as
\begin{equation}\label{neq59}
  h=f_0^2, \qquad a_2<f_0<a_1.
\end{equation}
Substituting this into \eqref{neq16}, we get
\begin{equation*}
  \ds{\left(\frac{d\xi}{d\tau}\right)^2=-\frac{(f_0-\lambda)^2}{\lambda},\qquad \left(\frac{d\eta}{d\tau}\right)^2=\frac{(f_0-\mu)^2}{\mu},}
\end{equation*}
so the trajectories corresponding to \eqref{neq59} are closed parametric curves $\lambda=f_0$ with different directions of motion. The manifold $\mfT_{h_0, f}$ consists of four isolated circles. These circles are pairwise identified when $f_0=a_2$ and $f_0=a_1$ (compare with \eqref{neq37}, \eqref{neq39}).

{\bf Non-critical motions.} According to Proposition~\ref{prop1} in the case when $(h, f)$ belongs to one of the open regions I -- IV (see Fig.~\ref{fig_01}), the manifold \eqref{neq18} consists of several two-dimensional tori with quasi-periodic motions on it; for the points outside the region bounded by the bifurcation diagram the manifold \eqref{neq18} is empty. Let us show what are the types of trajectories on the Poisson sphere of the unit vertical vector $m_1$ in these cases.

Denote by $\zeta_1<\zeta_2$ the solutions of the quadratic equation $\zeta^2-2f_0\zeta+h=0$. Then due to \eqref{neq26} and \eqref{neq27} for any pair $(h,f)$ the region of possible motions has the form
\begin{equation}\label{neq60}
  \tilde{M}_{h,f}=\{\zeta_1\ls\lambda\ls\zeta_2\}\cap(\{\mu\ls\zeta_2\}\cup\{\mu\gs\zeta_2\}).
\end{equation}

Suppose the point $(h,f)$ is in Region I. Then $a_3<\zeta_1<a_2<\zeta_2<a_1$ and according to inequalities \eqref{neq12} the region \eqref{neq60} has two connected components diffeomorphic to the square (see Fig.~\ref{fig_15}):
\begin{equation}\label{neq61}
  \tilde{M}_{h,f}=\{\lambda\ls\zeta_2, \mu\ls\zeta_1\}.
\end{equation}
The corresponding trajectories are superposition of two independent oscillations and fill \eqref{neq61} similar to Lissajous curves. The manifold \eqref{neq18} also has two connected components.

\begin{figure}[ht]
\begin{center}
\begin{minipage}[h]{0.4\linewidth}
\centering
\includegraphics[width=0.75\linewidth,keepaspectratio]{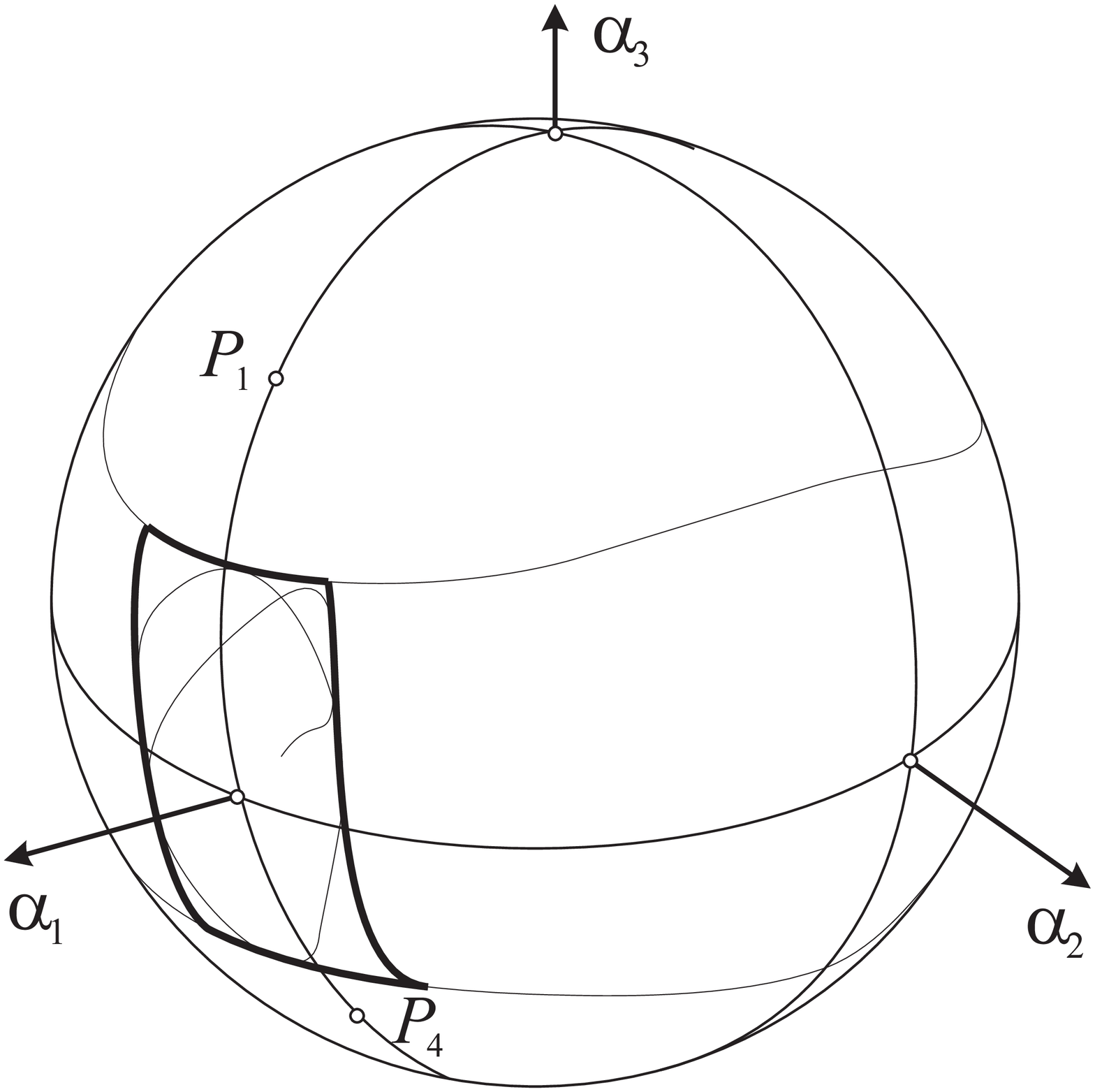}
\caption{}\label{fig_15}
\end{minipage}
\hspace{5mm}
\begin{minipage}[h]{0.4\linewidth}
\centering
\includegraphics[width=0.75\linewidth,keepaspectratio]{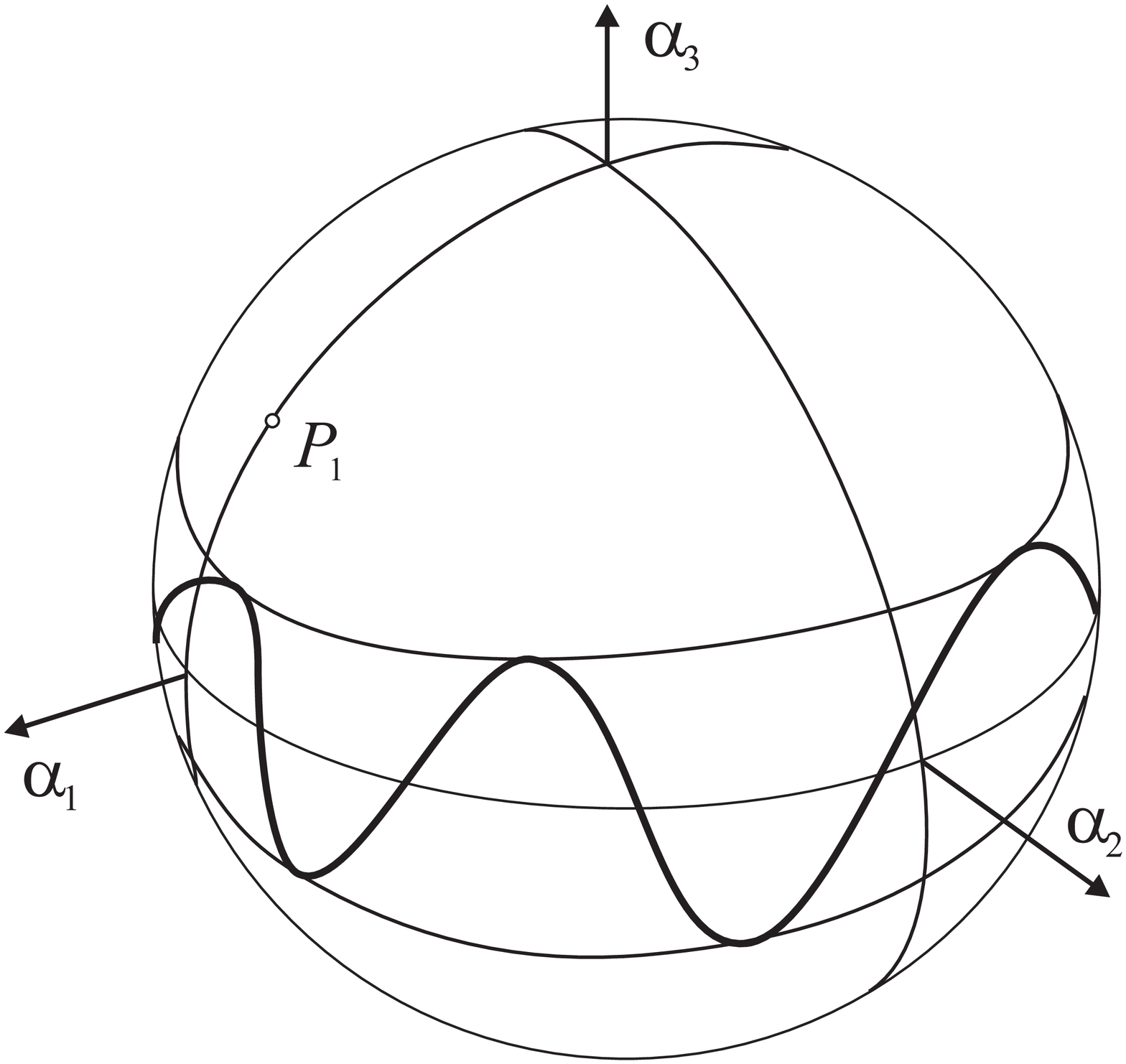}
\caption{}\label{fig_16}
\end{minipage}\\[3mm]
\begin{minipage}[h]{0.4\linewidth}
\centering
\includegraphics[width=0.75\linewidth,keepaspectratio]{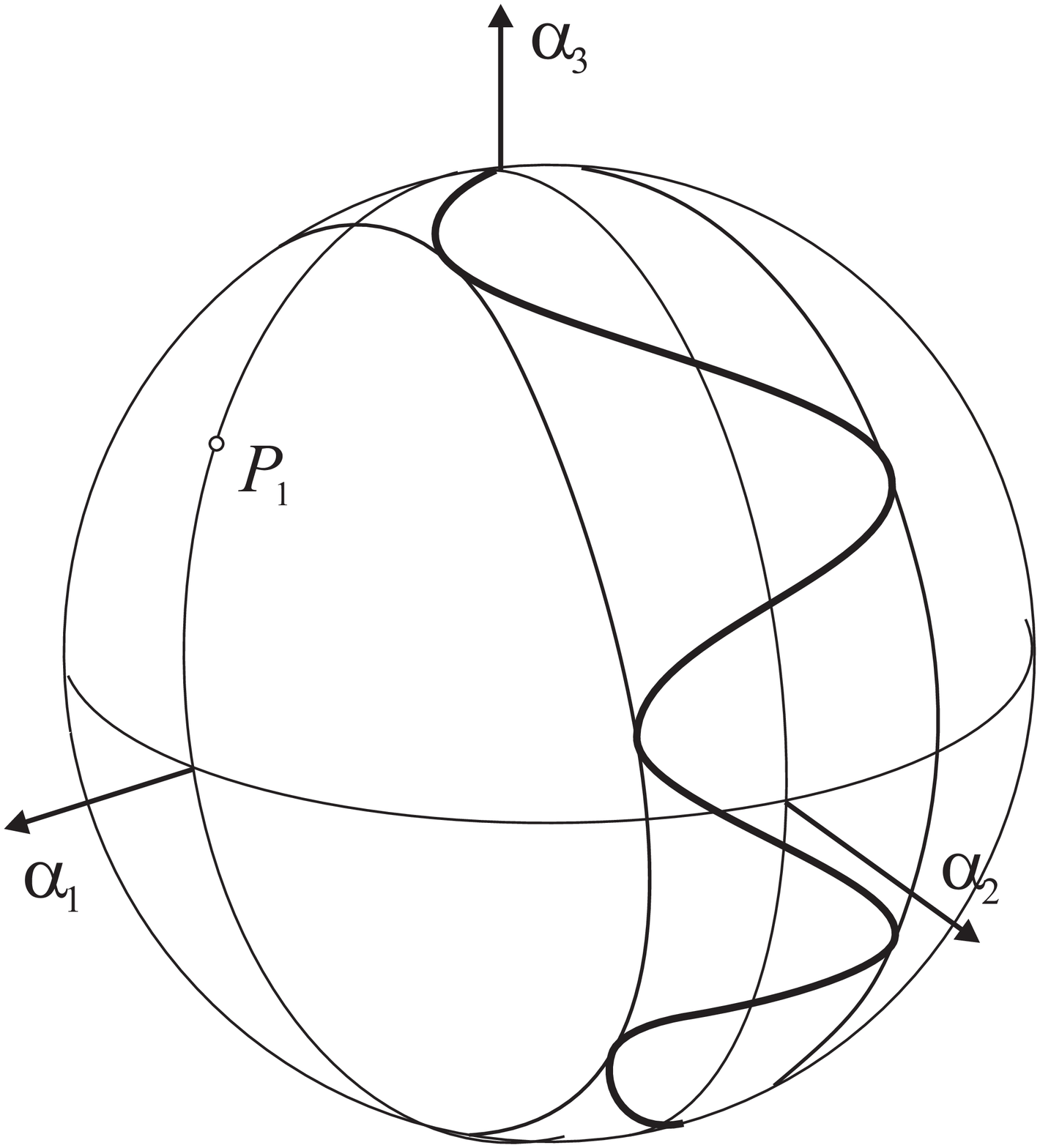}
\caption{}\label{fig_17}
\end{minipage}
\hspace{5mm}
\begin{minipage}[h]{0.4\linewidth}
\centering
\includegraphics[width=0.75\linewidth,keepaspectratio]{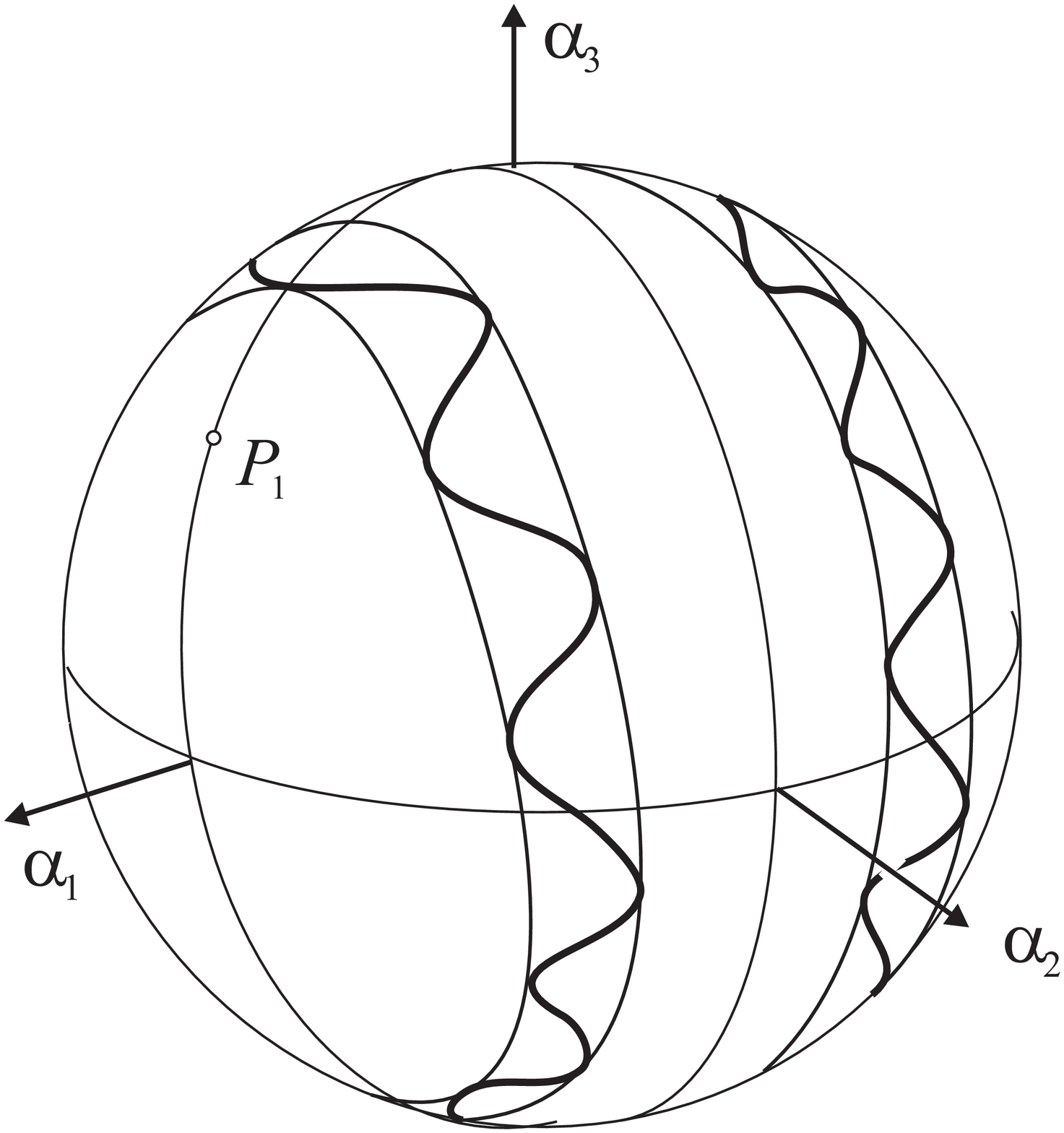}
\caption{}\label{fig_18}
\end{minipage}
\end{center}
\end{figure}

If $(h,f)$ belongs to Region II, then $\zeta_1$ still lies in the interval $(a_3, a_2)$, but $\zeta_2 > a_1$, therefore the region \eqref{neq60} is a ring
\begin{equation}\label{neq62}
  \tilde{M}_{h,f}=\{\mu\ls\zeta_1\},
\end{equation}
the trajectories in which go in both directions oscillating between its boundaries (see Fig.~\ref{fig_16}). When crossing the common border of Regions I and II the squares \eqref{neq61} glue together by ``vertical'' sides forming the ring \eqref{neq62}.

In the intermediate case, as it was shown before, this gluing generates limit cycles  (see Fig.~\ref{fig_03}).

In the case when $(h, f)$ is in Region III, the root $\zeta_1$ passes into the interval $(a_2, a_1)$, therefore the region \eqref{neq60} is a ring
\begin{equation}\label{neq63}
  \tilde{M}_{h,f}=\{\lambda\gs\zeta_1\},
\end{equation}
having the cross section $\alpha_1=0$ as its middle line. The motion in \eqref{neq63} is of the same type as in the previous case (see Fig.~\ref{fig_17}). The transformation of the region \eqref{neq60} when crossing the common border of Regions II and III is as follows: the segments of the cross section $\alpha_2=0$ between the points $P_1$, $P_2$ and $P_3$, $P_4$ are glued together, then we cut the segments of the same cross section between the points $P_1$, $P_4$ and $P_2$, $P_3$. In the cases \eqref{neq62} and \eqref{neq63} the manifold \eqref{neq18} has two connected components.

Finally let $(h, f)$ lie in Region IV. We have $a_2<\zeta_1<\zeta_2<a_1$, so the set \eqref{neq60} consists of two rings
\begin{equation}\label{neq64}
  \tilde{M}_{h,f}=\{\zeta_1\ls\lambda\ls\zeta_2\}.
\end{equation}
The motions in them occur in two directions with oscillation between the boundaries (see Fig.~\ref{fig_18}). The manifold \eqref{neq18} consists of four two-dimensional tori.

The pass from Region III into Region IV is trivial; the set \eqref{neq64} is obtained from \eqref{neq63} by cutting along the middle line. More interesting transformation takes place when crossing the common border of Regions I and IV. The upper and lower sides of the squares \eqref{neq61} are folded twice, thus, for example, the square \eqref{neq61} in the hemisphere $\alpha_1 > 0$ transforms into a disk containing the points $P_1$ and $P_4$, which after cutting along the segment of the cross section $\alpha_2=0$ between the points $P_1$ and $P_4$ generates one of the rings forming the set~\eqref{neq64}.

\end{document}